\documentclass[preprint,10pt]{sigplanconf}

\usepackage{listings}
\usepackage{epsfig}
\usepackage{amsmath}
\usepackage{mathpartir}
\usepackage{graphicx}
\usepackage{boxedminipage}
\usepackage{verbatim}
\usepackage{latexsym}
\usepackage{wrapfig}
\usepackage{alltt}

\newcommand{\dataflow}{dataflow}
\newcommand{\dc}{\mbox{\sc DC}}
\newcommand{\ceal}{\mbox{\sc CEAL}}
\newcommand{\qt}{\mbox{\tt Qt}}
\newcommand{\qtfs}{\mbox{\tt Qt} {\tt 4.6}}
\newcommand{\trans}{\,\Rightarrow}
\newcommand{\normal}{\mbox{\tt{normal}}}
\newcommand{\reactive}{\mbox{\tt{reactive}}}
\newcommand{\self}{\mbox{$c_{self}$}}
\newcommand{\pick}{\mbox{\tt{pick}}}
\newcommand{\newcons}{\mbox{\tt{newcons}}}
\newcommand{\delcons}{\mbox{\tt{delcons}}}
\newcommand{\batomic}{\mbox{\tt{begin\_at}}}
\newcommand{\eatomic}{\mbox{\tt{end\_at}}}

\newcommand{\while}{\mbox{\sc While}}
\newcommand{\dwhile}{\mbox{\sc DWhile}}
\newcommand{\wwhile}{\mbox{\bf while}}
\newcommand{\wdo}{\mbox{\bf do}}
\newcommand{\wskip}{\mbox{\bf skip}}
\newcommand{\wif}{\mbox{\bf if}}
\newcommand{\wthen}{\mbox{\bf then}}
\newcommand{\welse}{\mbox{\bf else}}
\newcommand{\emphasize}{\sc\small}

\newtheorem{theorem}{Theorem}
\newtheorem{lemma}{Lemma}

\newtheorem{definition}{Definition}

\def\hide#1{{}}

\newenvironment{proof}
{\noindent   {\sc Proof.}}{\hspace*{\fill}$\Box$\par\vspace{2mm}}

\newenvironment{sketch}
{\noindent   {\sc Proof (sketch).}}{\hspace*{\fill}$\Box$\par\vspace{2mm}}

\newcommand{\fullpaper}[1]{}
\newcommand{\wmax}[1]{}

\begin{document}

\conferenceinfo{WXYZ '05}{date, City.} 
\copyrightyear{2005} 
\copyrightdata{[to be supplied]} 

\title{Reactive Imperative Programming with Dataflow Constraints}

\authorinfo{Camil Demetrescu}
           {{\small Dept. of Computer and System Sciences\\Sapienza University of Rome}}
           {demetres@dis.uniroma1.it}
\authorinfo{Irene Finocchi}
           {{\small Dept. of Computer Science\\Sapienza University of Rome}}
           {finocchi@di.uniroma1.it}
\authorinfo{Andrea Ribichini}
           {{\small Dept. of Computer and System Sciences\\Sapienza University of Rome}}
           {ribichini@dis.uniroma1.it}

\maketitle

\begin{abstract}
\begin{small}
Dataflow languages provide natural support for specifying constraints between objects in dynamic applications, where programs need to react efficiently to changes of their environment. Researchers have long investigated how to take advantage of dataflow constraints by embedding them into procedural languages. Previous mixed imperative/dataflow systems, however, require syntactic extensions or libraries of {\em ad hoc} data types for binding the imperative program to the dataflow solver. In this paper we propose a novel approach that smoothly combines the two paradigms without placing undue burden on the programmer. 

In our framework, programmers can define ordinary commands of the host imperative language that enforce constraints between objects stored in special memory locations designated as ``reactive''. Differently from previous approaches, reactive objects can be of any legal type in the host language, including primitive data types, pointers, arrays, and structures. Commands defining constraints are automatically re-executed every time their input memory locations change, letting a program behave like a spreadsheet where the values of some variables depend upon the values of other variables. The constraint solving mechanism is handled transparently by altering the semantics of elementary operations of the host language for reading and modifying objects. We provide a formal semantics and describe a concrete embodiment of our technique into C/C++, showing how to implement it efficiently in conventional platforms using off-the-shelf compilers. We discuss common coding idioms and relevant applications to reactive scenarios, including incremental computation, observer design pattern, and data structure repair. The performance of our implementation is compared to {\em ad hoc} problem-specific change propagation algorithms, as well as to language-centric approaches such as self-adjusting computation and subject/observer communication mechanisms, showing that the proposed approach is efficient in practice.
\end{small}
\end{abstract}

\begin{small}




\end{small}

\section{Introduction}
\label{se:intro}

A one-way, dataflow constraint is an equation of the form $y=f(x_1,\ldots,x_n)$ in which the formula on the right side is automatically re-evaluated and assigned to the variable $y$ whenever any variable $x_i$ changes. If $y$ is modified from outside the constraint, the equation is left temporarily unsatisfied, hence the attribute ``one-way''. Dataflow constraints are recognized as a powerful programming methodology in a variety of contexts because of their versatility and simplicity~\cite{Zanden01}. The most widespread application of dataflow constraints is perhaps embodied by spreadsheets~\cite{AbrahamBE08, kay84}. In a spreadsheet, the user can specify a cell formula that depends on other cells: when any of those cells is updated, the value of the first cell is automatically recalculated. Rules in a {\tt makefile} are another example of dataflow constraints: a rule sets up a dependency between a target file and a list of input files, and provides shell commands for rebuilding the target from the input files. When the makefile is run, if any input file in a rule is discovered to be newer than the target, then the target is rebuilt. The dataflow principle can be also applied to software development and execution, where the role of a cell/file is replaced by a program variable. This approach has been widely explored in the context of interactive applications, multimedia animation, and real-time systems~\cite{Amulet97,WanH00,signal86, lustre87}.

Since the values of program variables are automatically recalculated upon changes of other values, the dataflow computational model is very different from the standard imperative model, in which the memory store is changed explicitly by the program via memory assignments. The execution flow of applications running on top of a dataflow environment is indeed data-driven, rather than control-driven, providing a natural ground for automatic change propagation  in all scenarios where programs need to react to modifications of their environment. Implementations of the dataflow principle share some common issues with self-adjusting computation, in which programs respond to input changes by updating automatically their output~\cite{AcarBBT06,HammerAC09,DBLP:conf/pldi/AcarBLTT10}.

Differently from purely declarative constraints~\cite{apt03principles}, data- flow constraints are expressed by means of (imperative) methods whose execution makes a relation satisfied. This programming style is intuitive and readily accessible to a broad range of developers~\cite{Zanden01}, since the ability to smoothly combine different paradigms in a unified framework makes it possible to take advantage of different programming styles in the context of the same application. The problem of integrating imperative and dataflow programming has already been the focus of previous work in the context of specific application domains~\cite{B81, myers90, Amulet97, Zanden01, 1640091}. Previous mixed imperative/dataflow systems are based on libraries of {\em ad hoc} data types and functions for representing constraint variables and for binding the imperative program to the constraint solver. One drawback of these approaches is that constraint variables can only be of special data types provided by the runtime library, causing loss of flexibility and placing undue burden on the programmer. A natural question is whether the dataflow model can be made to work with general-purpose, imperative languages, such as C, without adding syntactic extensions and {\em ad hoc} data types. In this paper we affirmatively answer this question.

\paragraph{Our Contributions.} We present a general-purpose framework where programmers can specify generic one-way constraints between objects of arbitrary types stored in {\em reactive} memory locations. Constraints are written as ordinary commands of the host imperative language and can be added and removed dynamically at run time. Since they can change multiple objects within the same execution, they are  multi-output. The main feature of a constraint is its sensitivity to modifications of reactive objects: a constraint is automatically re-evaluated whenever any of the reactive locations it depends on is changed, either by the imperative program, or by another constraint. A distinguishing feature of our approach is that the whole constraint solving mechanism is handled transparently by altering the semantics of elementary operations of the host imperative language for reading and modifying objects. No syntax extensions are required and no new primitives are needed except for adding/removing constraints, allocating/deallocating reactive memory locations, and controlling the granularity of solver activations. Differently from previous approaches, programmers are not forced to use any special data types provided by the language extension, and can resort to the full range of conventional constructs for accessing and manipulating objects offered by the host language. In addition, our framework supports all the other features that have been recognized to be important in the design of dataflow constraint systems~\cite{Zanden01}, including:

\begin{description}
\item {\em Arbitrary code}: constraints consist of arbitrary code that is legal in the underlying imperative language, thus including loops, conditionals, function calls, and recursion.

\item {\em Address dereferencing}: constraints are able to reference variables indirectly via pointers.

\item {\em Automatic dependency detection}: constraints automatically detect the reactive memory locations they depend on during their evaluation, so there is no need for programmers to explicitly declare dependencies, which are also allowed to vary over time.
\end{description}

\noindent We embodied these principles into an extension of C/C++ that we called \dc. Our extension has exactly the same syntax as C/C++, but a different semantics.  Our main contributions are reflected in the organization of the paper and can be summarized as follows:

\begin{itemize}
\item In Section~\ref{se:model} we abstract our mechanism showing how to extend an elementary imperative language to support one-way dataflow constraints using reactive memory. We distinguish between three main execution modes: normal, constraint, and scheduling. We formally describe our mixed imperative/dataflow computational model by defining the interactions between these modes and providing a formal semantics of our mechanism. 

\item In Section~\ref{ss:properties} we discuss convergence of the dataflow constraint solver by modeling the computation as an iterative process that aims at finding a common fixpoint for the current set of constraints. We identify general constraint properties that let the solver terminate and converge to a common fixpoint independently of the scheduling strategy. This provides a sound unifying framework for solving both acyclic and cyclic constraint systems. 

\item In Section~\ref{se:C} we describe the concrete embodiment of our technique into C/C++, introducing the main features of \dc.  \dc\ has exactly the same syntax as C/C++, but operations that read or modify objects have a different semantics. All other primitives, including creating and deleting constraints and allocating and deallocating reactive memory blocks, are provided as runtime library functions.

\item In Section~\ref{ss:examples} we give a variety of elementary and advanced programming examples and discuss how \dc\ can improve C/C++ programmability in three relevant application scenarios: incremental computation, implementation of  the observer software design pattern, and  data structure checking and repair. To the best of our knowledge, these applications have not been explored before in the context of dataflow programming.

\item In Section~\ref{se:implementation} we describe how \dc\ can be implemented using off-the-shelf compilers on conventional platforms via a combination of runtime libraries, hardware/operating system support, and dynamic code patching, without requiring any source code preprocessing.

\item In Section~\ref{se:experiments} we perform an extensive experimental analysis of  \dc\ in a variety of settings, showing that our implementation is effective in practice. We consider both interactive applications and computationally demanding benchmarks that manipulate lists, grids, trees, matrices, and graphs. We assess the performances of \dc\ against conventional C-based implementations as well as against competitors that can quickly react to input changes, i.e., {\em ad hoc} dynamic algorithms, incremental solutions realized in \ceal~\cite{HammerAC09} (a state-of-the-art C-based framework for self-adjusting computation), and \qt's signal-slot implementation of the subject/observer communication mechanism~\cite{QtBook06}.

\end{itemize}

\noindent Related work is discussed in Section~\ref{se:relatedwork} and directions for future research are sketched in Section~\ref{se:conclusions}. 

\section{Abstract Model}
\label{se:model}

To describe our approach, we consider an elementary imperative language and we show how to extend it to support one-way \dataflow\ constraints. We start from \while~\cite{DBLP:books/crc/CRCcompiler2002/PrasadA02}, an extremely simple language of commands including a sub-language of expressions. Although \while\ does not support many fundamental features of concrete imperative languages (including declarations, procedures, dynamic memory allocation, type checking, etc.), it provides all the building blocks for a formal description of our mechanism, abstracting away details irrelevant for our purposes. 
We discuss how to modify the semantics of \while\ to integrate a {\em \dataflow\ constraint solver}. We call the extended language \dwhile. \dwhile\ is identical to \while\ except for a different semantics and additional primitives for adding/deleting constraints dynamically and for controlling the granularity of solver activations. As we will see in Section~\ref{se:C}, these primitives can be supported in procedural languages as runtime library functions.

\subsection{The \dwhile\ Language}
\label{ss:dwhile-language}

The abstract syntax of \dwhile\ is shown in Figure~\ref{fi:dwhile-syntax}. The 
\begin{wrapfigure}{r}{3.8cm}
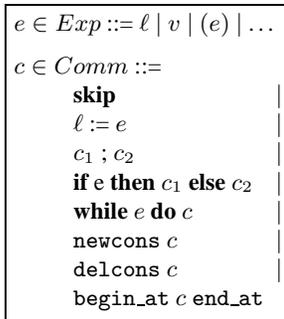

\begin{boxedminipage}{3.8cm}
\begin{small}
\noindent$e\in Exp$ ::= $\ell$ $|$ $v$ $|$ $(e)$ $|$ \ldots
\begin{tabbing}
xxxxx\=xxxxxxxxxxxxxxxxx\= \kill
$c\in Comm$ ::= \\ 
\> \wskip\ \> $|$ \\
\> $\ell$ := $e$ \> $|$  \\
\> $c_1$ ; $c_2$ \> $|$\\ 
\> \wif\ e \wthen\ $c_1$ \welse\ $c_2$ \> $|$\\
\> \wwhile\ $e$ \wdo\ $c$ \> $|$\\ 
\> \newcons\ $c$ \> $|$\\
\> \delcons\ $c$ \> $|$\\
\> \batomic\ $c$ \eatomic\ \>
\end{tabbing}
\end{small}
\end{boxedminipage}
\vspace{-2mm}
\nocaptionrule\caption{Abstract syntax of \dwhile.}
\label{fi:dwhile-syntax}
\end{wrapfigure}
language distinguishes between commands and expressions. We use 
$c$, $c_1$, $c_2$ as 
meta-variables ranging over the set of commands $Comm$, and $e$, $e_1$, $e_2$ as meta-variables 
ranging over the set of expressions $Exp$. Canonical forms
of expressions are either storage locations $\ell\in Loc$, or storable values $v$ over some arbitrary domain $Val$. Expressions can be also obtained by applying  to sub-expressions any primitive operations defined over domain $Val$ (e.g., plus, minus, etc.). Commands include:

\begin{itemize}

\item Assignments of values to storage locations ($\ell$ := $e$). These commands are the basic state transformers.

\item Constructs for sequencing, conditional execution, and iteration, with the usual meaning.

\item Two new primitives, $\newcons$ and $\delcons$, for adding and deleting constraints dynamically. Notice that a constraint in \dwhile\ is just an ordinary command. 

\item An atomic block construct, {\tt begin\_at c \tt end\_at}, that executes a command {\tt c} atomically so that any constraint evaluation is deferred until the end of the block. This offers fine-grained control over solver activations.

\end{itemize}

\noindent In Section~\ref{se:C} we will show a direct application of the concepts developed in this section to the C/C++ programming languages.

\subsection{Memory Model and Execution Modes}
\label{ss:modes}

Our approach hinges upon two key notions: {\em reactive memory locations} and {\em constraints}. Reactive memory can be read and written just like ordinary memory. However, differently from ordinary memory:

\begin{enumerate}

\item If a constraint $c$ reads a reactive memory location $\ell$ during its execution, a dependency $(\ell,c)$ of $c$ from $\ell$ is logged in a set $D$ of dependencies.

\item If the value stored in a reactive memory location $\ell$ is changed, all constraints depending on $\ell$ (i.e., all constraints $c$ such that $(\ell,c)\in D$) are automatically re-executed.
\end{enumerate}

\noindent Point 2 states that constraints are sensitive to modifications of the contents of the reactive memory. Point 1 shows how to maintain dynamically the set $D$ of dependencies needed to trigger the appropriate constraints upon changes of reactive memory locations. We remark that re-evaluating a constraint $c$ may completely change the set of its dependencies: prior to re-execution, all the old dependencies $(-,c)\in D$ are discarded, and new dependencies are logged in $D$ during the re-evaluation of $c$.

\begin{figure}
\centerline{\includegraphics[width=8.3cm]{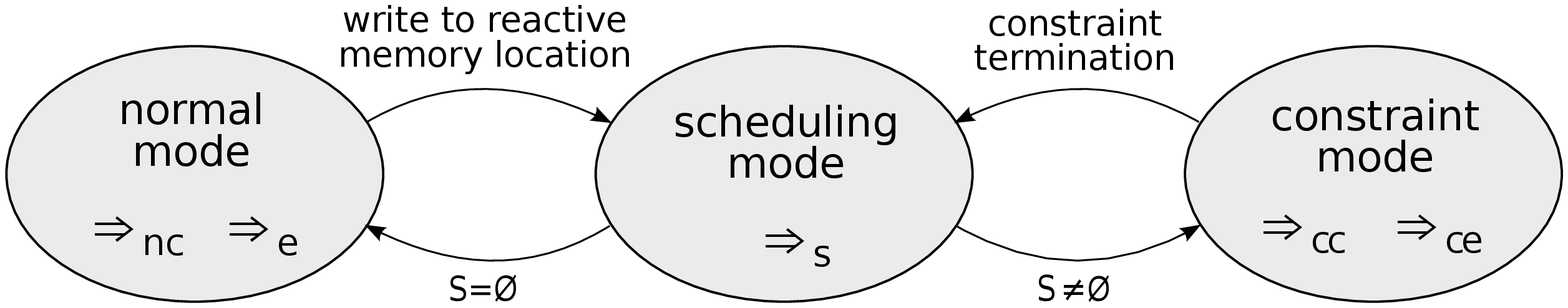}}
\bigskip
\caption{Transitions between different execution modes.}
\label{fi:automa}
\end{figure}

As shown in Figure~\ref{fi:automa}, at any point in time the execution can be in one of three modes: {\em normal execution}, {\em constraint execution}, or {\em scheduling}. As we will see more formally later in this section, different instructions (such as reading a reactive memory location or assigning it with a value) may have different semantics depending on the current execution mode.

We assume {\em eager} constraint evaluation, i.e., out-of-date constraints are brought up-to-date as soon as possible. This choice is better suited to our framework and, as previous experience has shown, lazy and eager evaluators typically deliver comparable performance in practice~\cite{Zanden01}.
Eager evaluation is achieved as follows. A scheduler maintains a data structure $S$ containing constraints to be first executed or re-evaluated. As an invariant property, $S$ is guaranteed to be empty during normal execution. As soon as a reactive memory location $\ell$ is written, the scheduler queries the set $D$ of dependencies and adds to $S$ all the constraints depending on $\ell$. These constraints are then run one-by-one in constraint execution mode, and new constraints may be added to $S$ throughout this process. Whenever $S$ becomes empty, normal execution is resumed. 

An exception to eager evaluation is related to atomic blocks. The execution of an atomic block {\tt c} is regarded as an uninterruptible operation: new constraints created during the evaluation of {\tt c} are just added to $S$. When {\tt c} terminates, for each reactive memory location $\ell$ whose value has changed, all the constraints depending on $\ell$ are also added to $S$, and the solver is eventually activated. Constraint executions are uninterruptible as well.

We remark that any scheduling mechanism may be used for selecting from $S$ the next constraint to be evaluated: in this abstract model we rely on a function $\pick$ that implements any appropriate scheduling strategy. 

\begin{figure*}
\begin{boxedminipage}{\linewidth}
\begin{small}
\begin{tabular}{p{10.8cm}p{5.8cm}}

$\trans\,\subseteq ({\cal R}\times \Sigma\times Comm)\times \Sigma$
\vspace{-3mm}
&
${\langle \rho,\sigma, c\rangle\trans \langle \sigma'\rangle}$
\vspace{-3mm}

\\

$\trans_{ce}\,\subseteq ({\cal R}\times\Sigma\times Cons\times Dep\times Exp)\times (Dep\times Val)$
\vspace{-3mm}
&
${\langle \rho,\sigma, c_{self}, D, e\rangle\trans_{ce} \langle D', v\rangle}$
\vspace{-3mm}

\\

$\trans_{nc}\,\subseteq ({\cal R}\times Bool\times\Sigma\times Dep\times 2^{Cons}\times Comm)\times (\Sigma\times Dep\times 2^{Cons})$
\vspace{-3mm}
&
${\langle \rho,a,\sigma, D, S, c\rangle\trans_{nc} \langle \sigma', D', S'\rangle}$
\vspace{-3mm}

\\

$\trans_{cc}\,\subseteq ({\cal R}\times\Sigma\times Dep\times 2^{Cons}\times Cons\times Comm)\times (\Sigma\times Dep\times 2^{Cons})$
\vspace{-3mm}
&
${\langle \rho,\sigma, D, S, c_{self}, c\rangle\trans_{cc} \langle \sigma', D', S'\rangle}$
\vspace{-3mm}

\\

$\trans_{s}\,\subseteq ({\cal R}\times \Sigma\times Dep\times 2^{Cons})\times (\Sigma\times Dep)$
\vspace{-3mm}
&
${\langle \rho,\sigma, D, S\rangle\trans_{s} \langle \sigma', D'\rangle}$
\vspace{-3mm}

\end{tabular}
\end{small}
\end{boxedminipage}
\vspace{0mm}
\nocaptionrule\caption{Transition relations for \dwhile\ program evaluation ($\trans$), expression evaluation in constraint mode ($\trans_{ce}$), command execution in normal mode ($\trans_{nc}$), command execution in constraint mode ($\trans_{cc}$), and constraint solver execution in scheduling mode ($\trans_{s}$).}
\label{fi:transitions}
\end{figure*}

\subsection{Configurations}
\label{ss:configurations}

A configuration of our system is a six-tuple 
\vspace{-1mm}
$$(\rho, a, \sigma, D, S, \self)\in {\cal R}\times Bool\times\Sigma\times Dep\times 2^{Cons}\times Cons\vspace{-1mm}$$
where:
\begin{itemize}

\item ${\cal R} = \{\rho: Loc \rightarrow \{\,\normal,\reactive\,\} \}$ is a set of store attributes, i.e., Boolean functions specifying which memory locations are reactive.

\item $Bool=\{\mbox{\tt true}, \mbox{\tt false}\}$ is the set of Boolean values.

\item $\Sigma = \{\sigma: Loc \rightarrow Val \}$ is a set of stores mapping storage locations to storable values. 

\item $Cons$ is the set of constraints and $2^{Cons}$ denotes its power set. A constraint can be any command in \dwhile, i.e., $Cons=Comm$. We use different names for the sake of clarity.

\begin{figure*}
\begin{boxedminipage}{\linewidth}
\begin{small}

\begin{tabular}{p{8.3cm}|p{8.3cm}}

\begin{minipage}{8.3cm}
\begin{center}

\begin{math}
\begin{array}{ll}
\inferrule 
{\rho, a, S\vdash \langle \sigma, D, c\rangle\trans_{nc} \langle \sigma', D'\rangle}
{\rho \vdash \langle \sigma, c\rangle \trans \sigma'}
&
\raisebox{1pt}
{
\begin{math}
\mbox{where: } 
\left\{ 
\begin{array}{l}
a=\mbox{\tt false}\\
D=\emptyset \\
S=\emptyset \\
\end{array} 
\right. 
\end{math}
} 
\end{array}
\end{math}

\end{center}
\end{minipage}

&

\begin{minipage}{8.3cm}
\begin{center}
\begin{math}
\inferrule 
{\rho,\sigma,c_{self}\vdash\langle D, e\rangle\trans_{ce} \langle D',v\rangle ~~~~~~ \sigma'=\sigma|_{\ell\mapsto v}}
{\rho, S, c_{self}\vdash\langle\sigma, D, \ell:=e\rangle\trans_{cc} \langle \sigma', D'\rangle}
\end{math}
\end{center}
\end{minipage}

\\
\centerline{(\sc Eval)}
\vspace{-3mm}
&

\centerline{(\sc Asgn-c)}
\vspace{-3mm}

\\\hline

\begin{minipage}{8.3cm}
\begin{center}
\begin{math}
\inferrule 
{\sigma\vdash e\trans_{e} v ~~~~~~ \sigma'=\sigma|_{\ell\mapsto v}}
{\rho, a, D, S\vdash\langle\sigma, \ell:=e\rangle\trans_{nc} \sigma'}
\end{math}
\end{center}
\end{minipage}

&

\begin{minipage}{8.3cm}
\begin{center}
\begin{math}
\inferrule 
{S=\emptyset ~~~~~ S'=\{c\,|\,(\ell,c)\in D\} \\ \sigma\vdash e \trans_{e} v ~~~~~ \sigma'=\sigma|_{\ell\mapsto v} ~~~~~ 
\rho\vdash \langle\sigma', D, S'\rangle\trans_s \langle\sigma'', D'\rangle}
{\rho, a, S\vdash\langle\sigma, D, \ell:=e\rangle\trans_{nc} \langle\sigma'', D'\rangle}
\end{math}
\end{center}
\end{minipage}

\\
\centerline{$\mbox{if } \rho(\ell)=\normal\ \mbox{or } \sigma'(\ell)=\sigma(\ell) \mbox{ or } a=\mbox{\tt true}$}
\vspace{-3mm}
&
\centerline{$\mbox{if } \rho(\ell)=\reactive\ \mbox{and } \sigma'(\ell)\neq\sigma(\ell) \mbox{ and } a=\mbox{\tt false}$}
\vspace{-3mm}
\\

\centerline{(\sc Asgn-n1)}
\vspace{-3mm}
&

\centerline{(\sc Asgn-n2)}
\vspace{-3mm}

\\\hline

\begin{minipage}{8.3cm}
\begin{center}
\begin{math}
\begin{array}{cc}
\inferrule 
{\sigma\vdash \ell\trans_e v}
{\rho, \sigma, c_{self}, D\vdash \ell\trans_{ce}v}
&
\mbox{ if } \rho(\ell)=\normal
\end{array}
\end{math}
\end{center}
\end{minipage}

&

\begin{minipage}{8.3cm}
\begin{center}
\begin{math}
\begin{array}{cc}
\inferrule 
{\sigma\vdash \ell\trans_e v \\
D'=D\cup\{(\ell,\self)\}}
{\rho, \sigma, c_{self}\vdash \langle D, \ell\rangle\trans_{ce}\langle D', v\rangle}
&
\mbox{ if } \rho(\ell)=\reactive
\end{array}
\end{math}
\end{center}
\end{minipage}

\\

\centerline{(\sc Deref-c1)}
\vspace{-3mm}
&

\centerline{(\sc Deref-c2)}
\vspace{-3mm}

\\\hline


\begin{minipage}{8.3cm}
\begin{center}
\begin{math}
\inferrule
{\rho, c_{self} \vdash \langle \sigma, D, S, c\rangle\trans_{cc}\langle \sigma', D', S'\rangle}
{\rho, c_{self} \vdash \langle \sigma, D, S, \batomic\ c~\eatomic\rangle\trans_{cc}\langle \sigma', D', S'\rangle}
\end{math}
\end{center}
\end{minipage}

&

\begin{minipage}{8.3cm}
\begin{center}
\begin{math}
\begin{array}{l@{~~}l}
\inferrule
{\rho, a, D \vdash \langle \sigma, S, c\rangle\trans_{nc}\langle \sigma', S'\rangle}
{\rho, a, D \vdash \langle \sigma, S, \batomic\ c~\eatomic\rangle\trans_{nc}\langle \sigma', S'\rangle}
&
\mbox{if } a=\mbox{\tt true}
\end{array}
\end{math}
\end{center}
\end{minipage}

\\

\centerline{(\sc BeginEnd-c)}
\vspace{-3mm}

&

\centerline{(\sc BeginEnd-n1)}
\vspace{-3mm}
\\\hline

\multicolumn{2}{c}{
\begin{math}
\begin{array}{l@{~~~~~}l}
\inferrule
{S=\emptyset ~~~~~~~~ a'=\mbox{\tt true} ~~~~~~~~ \rho, D\vdash \langle a', \sigma, S, c\rangle\trans_{nc} \langle\sigma', S'\rangle \\\\ 
S''=S'\cup \{\,c\,|\,(\ell,c)\in D \wedge \sigma(\ell)\neq \sigma'(\ell) \wedge\,\rho(\ell)=\reactive\,\} ~~~~~~~~
\rho\vdash \langle \sigma', D, S''\rangle\trans_{s}\langle \sigma'', D'\rangle}
{\rho, a, S\vdash \langle \sigma, D, \batomic\ c~\eatomic\rangle\trans_{nc}\langle \sigma'', D'\rangle}
&
\mbox{if } a=\mbox{\tt false}
\end{array}
\end{math}
}
\smallskip\\
\multicolumn{2}{c}{(\sc BeginEnd-n2)\smallskip }

\\\hline

\begin{minipage}{8.3cm}
\begin{center}
\begin{math}
\begin{array}{l@{~\,\,}l}
\hspace{-2mm}\inferrule
{ S=\emptyset ~~~~~S'=\{c\} ~~~~~ \rho\vdash \langle \sigma, D, S'\rangle\trans_s \langle\sigma', D'\rangle}
{\rho, a, S \vdash\langle \sigma, D, \newcons \mbox{\tt { c}}\rangle\trans_{nc} \langle\sigma', D'\rangle}
& 
\mbox{if } a=\mbox{\tt false}
\end{array}
\end{math}
\end{center}
\end{minipage}

&

\begin{minipage}{8.3cm}
\begin{center}
\begin{math}
\begin{array}{ll}
\inferrule
{ S' = S\cup\{c\}}
{\rho, a, \sigma, D  \vdash\langle S, \newcons \mbox{\tt { c}}\rangle\trans_{nc} S'}
&
\mbox{if } a=\mbox{\tt true}
\end{array}
\end{math}
\end{center}
\end{minipage}

\\

\centerline{(\sc NewCons-n1)}
\vspace{-3mm}

&

\centerline{(\sc NewCons-n2)}
\vspace{-3mm}

\\\hline

\begin{minipage}{8.3cm}
\begin{center}
\begin{math}
\inferrule
{ D'=D\setminus\{(-,c)\} ~~~~~ S'=S \setminus\{c\}}
{\rho, a, \sigma\vdash\langle D, S, \delcons \mbox{\tt { c}}\rangle\trans_{nc} \langle D', S'\rangle}
\end{math}
\end{center}
\end{minipage}

&

\begin{minipage}{8.3cm}
\begin{center}
\begin{math}
\inferrule
{ S' = S\cup\{c\}}
{\rho, \sigma, D, c_{self} \vdash\langle S, \newcons \mbox{\tt { c}}\rangle\trans_{cc} S'}
\end{math}
\end{center}
\end{minipage}

\\

\centerline{(\sc DelCons-n)}
\vspace{-3mm}

&

\centerline{(\sc NewCons-c)}
\vspace{-3mm}

\\\hline

\begin{minipage}{8.3cm}
\begin{center}
\begin{math}
\inferrule
{ D'=D\setminus\{(-,c)\} ~~~~~ S'=S \setminus\{c\}}
{\rho, \sigma, c_{self} \vdash\langle D, S, \delcons \mbox{\tt { c}}\rangle\trans_{cc} \langle D', S'\rangle}
\end{math}
\end{center}
\end{minipage}

&

\begin{minipage}{8.3cm}
\begin{center}
\begin{math}
\begin{array}{ll}
\inferrule
{~}
{\rho\vdash \langle \sigma, D, S\rangle\trans_{s}\langle \sigma, D\rangle}
& \mbox{ if } S=\emptyset
\end{array}
\end{math}
\end{center}
\end{minipage}

\\

\centerline{(\sc DelCons-c)}
\vspace{-3mm}

&

\centerline{(\sc Solver-1)}
\vspace{-3mm}

\\\hline


\multicolumn{2}{c}{
\begin{math}
\begin{array}{lll}
\inferrule
{\rho\vdash \langle\sigma, D', S\setminus\{\self\}, \self, \self\rangle\trans_{cc}\langle \sigma', D'', S'\rangle \\\\
\rho\vdash \langle\sigma', D'', S''\rangle\trans_s \langle\sigma'', D'''\rangle
}
{\rho\vdash \langle \sigma, D, S\rangle\trans_{s}\langle \sigma'', D'''\rangle
}
&
\raisebox{6pt}
{
where:
\begin{math}
\left\{\begin{array}{l}
\self=\pick(S)\\
D' = D\setminus\{(-\,,\self)\} \\
S''= S'\cup \{ c\,|\,(\ell,c)\in D''\,\wedge\, \\
~~~~~~~~~~\, \sigma(\ell)\neq\sigma'(\ell)\,\wedge\,\rho(\ell)=\reactive \}
\end{array}\right .
\end{math}
}
&
\raisebox{6pt}{$\mbox{ if } S\neq\emptyset$}
\end{array}
\end{math}
} 
\smallskip\\
\multicolumn{2}{c}{ (\sc Solver-2)\smallskip }

\end{tabular}

\end{small}
\end{boxedminipage}
\vspace{0mm}
\nocaptionrule\caption{\dwhile\ program evaluation.}
\label{fi:semantics}
\end{figure*}


\item $Dep=2^{Loc \times Cons}$ is the set of all subsets of dependencies of constraints from reactive locations. 

\end{itemize}

\noindent Besides a store $\sigma$ and its attribute $\rho$, a configuration includes: 

\begin{itemize}
\item a Boolean flag $a$ that is {\tt true} inside atomic blocks and is used for deferring solver activations; 
\item the set $D$ of dependencies, $D\subseteq Loc \times Cons$; 
\item the scheduling data structure $S\subseteq Cons$ discussed above;
\item a meta-variable $c_{self}$ that denotes the {\em current} constraint (i.e., the constraint that is being evaluated) in constraint execution mode, and is undefined otherwise. If the scheduler were deterministic, $\self$ may be omitted from the configuration, but we do not make this assumption in this paper.
\end{itemize}

\subsection{Operational Semantics}
\label{ss:semantics}

Most of the operational semantics of the \dwhile\ language can be directly derived from the standard semantics of \while. The most interesting aspects of our extension include reading and writing the reactive memory, adding and deleting constraints, excuting commands atomically, and defining the behavior of the scheduler and its interactions with the other execution modes. Rules for these aspects are given in Figure~\ref{fi:semantics} and are discussed below.

Let $\trans_e\,\subseteq (\Sigma\times Exp)\times Val$ and $\trans_c\,\subseteq (\Sigma\times Comm)\times \Sigma$ be the standard big-step transition relations used in the operational semantics of the \while\ language~\cite{DBLP:books/crc/CRCcompiler2002/PrasadA02}. Besides  $\trans_e$ and $\trans_c$, we use additional transition relations for expression evaluation in constraint mode ($\trans_{ce}$), command execution in normal mode ($\trans_{nc}$), command execution in constraint mode ($\trans_{cc}$), and constraint solver execution in scheduling mode ($\trans_{s}$), as defined in Figure~\ref{fi:transitions}.
Notice that expression evaluation in normal mode can be carried on directly by means of transition relation $\trans_e$ of \while. As discussed below, relation $\trans_{ce}$ is obtained by appropriately modifying $\trans_e$. Similarly, relations $\trans_{nc}$ and $\trans_{cc}$ are obtained by appropriately modifying $\trans_c$. All the rules not reported in Figure~\ref{fi:semantics} can be derived in a straightforward way from the corresponding rules in the standard semantics of \while~\cite{DBLP:books/crc/CRCcompiler2002/PrasadA02}.

The evaluation of a \dwhile\ program is started by rule {\sc Eval}, which initializes the atomic flag $a$ to {\tt false} and both the scheduling queue $S$ and the set $D$ of dependencies to the empty set.

\paragraph{Writing Memory.} Assigning an ordinary memory location in normal execution mode (rule {\sc Asgn-n1}) just changes the store as in the usual semantics of \while. This is also the case when the new value of the location to be assigned equals its old value or inside an atomic block.
Otherwise, if the location $\ell$ to be assigned is reactive, the new value differs from the old one, and execution is outside atomic blocks (rule {\sc Asgn-n2}), constraints depending on $\ell$ are scheduled in  $S$ and are evaluated one-by-one. As we will see, the transition relation $\trans_{s}$ guarantees $S$ to be empty at the end of the constraint solving phase. In conformity with the atomic execution of constraints, assignment in constraint mode (rule {\sc Asgn-c}) just resorts to ordinary assignment in \while\ for both normal and reactive locations. We will see in rule {\sc Solver-2}, however, that constraints can be nevertheless scheduled by other constraints if their execution changes the contents of reactive memory locations.

\paragraph{Reading Memory.} Reading an ordinary memory location in constraint execution mode (rule {\sc Deref-c1}) just evaluates the location to its value in the current store: this is achieved by using transition relation $\trans_e$ of the \while\ semantics. If the location $\ell$ to be read is reactive (rule {\sc Deref-c2}), a new dependency of the active constraint $c_{self}$ from $\ell$ is also added to the set $D$ of dependencies.

\paragraph{Executing Atomic Blocks.} To execute an atomic block in normal mode (rule {\sc BeginEnd-n2}), the uninterruptible command $c$ is first evaluated according to the rules defined by transition $\trans_{nc}$. If the content of some reactive location changes due to the execution of $c$, the solver is then activated at the end of the block. The {\tt begin\_at}\,/\,{\tt end\_at} command has instead no effect when execution is already atomic, i.e., in constraint mode (rule {\sc BeginEnd-c}) and inside atomic blocks (rule {\sc BeginEnd-n1}), except for executing command $c$.

\paragraph{Creating and Deleting Constraints.} In non-atomic normal execution mode, rule {\sc NewCons-n1} creates a new constraint and triggers its first execution by resorting to $\trans_s$. In atomic normal execution and in constraint mode,  rules {\sc NewCons-n2}  and {\sc NewCons-c} simply add the constraint to the scheduling queue. Similarly, rules {\sc DelCons-n} and {\sc DelCons-c} remove the constraint from the scheduling queue and clean up its dependencies from $D$.   
 
\paragraph{Activating the Solver.} Rules {\sc Solver-1} and {\sc Solver-2} specify the behavior of the scheduler, which is started by rules {\sc Asgn-n2} and {\sc BeginEnd-n2}. Rule {\sc Solver-1} defines the termination of the constraint solving phase: this phase  ends only when there are no more constraints to be evaluated (i.e., $S=\emptyset$). Rule  {\sc Solver-2} has an inductive definition. If $S$ is not empty, function $\pick$ selects from $S$ a new active constraint $c_{self}$, which is evaluated in constraint mode after removing from $D$ its old dependencies. The final state ($\sigma''$) and dependencies ($D'''$) are those obtained by applying the scheduler on the store $\sigma'$ obtained after the execution of $c_{self}$ and on a new set $S''$ of constraints. $S''$ is derived from $S$ by adding
any new constraints ($S'$) resulting from the execution of $c_{self}$ along with the constraints depending on reactive memory locations whose content has been changed by $c_{self}$. The definition of $S''$ guarantees that constraints can trigger other constraints (even themselves), even if each constraint execution is regarded as an atomic operation and is never interrupted by the scheduler.

\section{Convergence Properties}
\label{ss:properties}

In this section, we discuss some general properties of the constraint solving mechanism we adopt in DC, including termination, correctness, and running times. The computation of one-way dataflow constraints (similarly to spreadsheet formulas, circuits, etc.) is traditionally described in the literature in terms of a bipartite directed graph called {\em \dataflow\ graph}. In a \dataflow\ graph, a node can model either an execution unit (e.g., gate~\cite{AlpernHRSZ90}, process~\cite{systemcbook02}, one-way constraint~\cite{Zanden01}, or spreadsheet formula~\cite{kay84}) or an input/output port of one or more units (e.g., gate port, variable, or cell). There is an arc from a port to an execution unit if the unit uses that port as a parameter, and from an execution unit to a port if the unit assigns a value to that port. Paths in a dataflow graph, which is usually acyclic, describe how data flows through the system, and the result of a computation can be characterized algorithmically in terms of an appropriate traversal of the graph (e.g., in a topological order). This model is very effective in describing scenarios where data dependencies are either specified explicitly, or can be derived statically from the program. However, in general the \dataflow\ graph might be not known in advance or may evolve over time in a manner that may be difficult to characterize. In all such cases, proving general properties of programs based on the evaluation of the \dataflow\ graph may not be easy. A more general approach, which we follow in our work, consists of modeling dataflow constraint solving as an iterative process that aims at finding a common fixpoint for the current set of constraints. In our context, a fixpoint is a store that satisfies simultaneously all the relations between reactive memory locations specified by the constraints. 
This provides a unifying framework for solving dataflow constraint systems with both acyclic and cyclic dependencies.

\subsection{Independence of the Scheduling Order} 
\label{ss:termination}

In Section~2, we have assumed that the scheduling order of constraint executions is specified by a function $\pick$ given as a parameter of the solver. A natural question is whether there are any general properties of a set of constraints that let our solver terminate and converge to a common fixpoint independently of the scheduling strategy used by function $\pick$. Using results from the theory of function iterations~\cite{DBLP:conf/popl/CousotC77}, we show that any arbitrary collection of inflationary one-way constraints has the desired property. This class of constraints includes, for instance, any program that can be described in terms of an acyclic dataflow graph such as computational circuits~\cite{AlpernHRSZ90}, non-circular attribute grammars~\cite{knuth68}, and spreadsheets~\cite{kay84} (see Section~\ref{ss:constraint-examples}).  We remark, however, that it is more general as it allows it to address problems that would not be solvable without cyclic dependencies (an example is given in Section~\ref{ss:dynamic-sssp}).

We first provide some preliminary definitions in accordance with the terminology used in~\cite{apt03principles}. We model constraint execution as the application of functions on stores:

\begin{definition}
We denote by $f_c:\Sigma\rightarrow\Sigma$ the function computed by a constraint $c\in Cons$, where $f_c(\sigma)=\sigma'$ if $\langle\sigma,c\rangle\trans_c\sigma'$. We say that store $\sigma\in\Sigma$ is a {\emphasize fixpoint} for $f_c$ if $f_c(\sigma)=\sigma$.
\end{definition}

\noindent To simplify the discussion, throughout this section we assume that constraints only operate on reactive cells and focus our attention on stores where all locations are reactive. The definition of inflationary functions assumes that a partial ordering is defined on the set of stores $\Sigma$:

\begin{definition}[\emphasize Inflationary Functions]
Let $(\Sigma, \preceq)$ be any partial ordering over the set of stores $\Sigma$ and let $f:\Sigma\rightarrow\Sigma$ be a function on $\Sigma$. We say that $f$ is {\em inflationary} if $\sigma \preceq f(\sigma)$ for all $\sigma\in\Sigma$. 
\end{definition}

\noindent Examples of partial orderings on $\Sigma$ will be given in Section~\ref{ss:constraint-examples} and in Section~\ref{ss:cyclic}. A relevant property of partial orderings in our context is the {\em finite chain property}, based on the notion of {\em sequence stabilization}:

\begin{definition}[\emphasize Finite Chain Property]
A partial ordering $(\Sigma,$ $\preceq)$ over $\Sigma$ satisfies the {\emphasize Finite Chain Property} if every non-decreasing sequence of elements $\sigma_0 \preceq\sigma_1 \preceq\sigma_2 \preceq\ldots$ \,from $\Sigma$ eventually stabilizes at some element $\sigma$ in $\Sigma$, i.e., if there exists $j\geq 0$ such that $\sigma_i=\sigma$ for all $i\geq j$.
\end{definition}

\noindent To describe the store modifications due to the execution of the solver, we use the notion of {\em iteration of functions} on stores. Let $F=\{f_1,\ldots,f_n\}$, $\langle a_1, \ldots, a_k\rangle$, and $\sigma\in\Sigma$ be a finite set of functions on $\Sigma$, a sequence of indices in $[1,n]$, and an initial store, respectively. An iteration of functions of $F$ starting at $\sigma$ is a sequence of stores $\langle\sigma_0, \sigma_1, \sigma_2,  \ldots\rangle$ where $\sigma_0=\sigma$ and $\sigma_{i}=f_{a_{i}}(\sigma_{i-1})$ for $i>0$. We say that function $f_{a_{i}}$ is {\em activated} at step $i$. Iterations of functions that lead to a fixed point are called {\em regular}:

\begin{definition}[\emphasize Regular Function Iteration]
~~A function iteration $\langle\sigma_0, \sigma_1, \sigma_2,  \ldots\rangle$ is {\emphasize regular} if it satisfies the following property: for all $f\in F$ and $i\ge 0$, if $\sigma_i$ is not a fixpoint for $f$, then $f$ is activated at some step $j>i$.
\end{definition}

\noindent Using arguments from Chapter 7 of~\cite{apt03principles}, it can be proved that any regular iteration of inflationary functions starting at some initial store stabilizes in a finite number of steps to a common fixpoint:

\begin{lemma}[\emphasize Fixpoint]
\label{le:fixpoint}
Let $(\Sigma, \preceq)$ be any partial ordering over $\Sigma$ satisfying the finite chain property and let $F$ be a finite set of inflationary functions on $\Sigma$. Then any regular iteration of $F$ starting at $\sigma$ eventually stabilizes at a common fixpoint $\sigma'$ of the functions in $F$ such that $\sigma \preceq \sigma'$.
\end{lemma}

\noindent We can now discuss convergence properties of our solver:

\begin{theorem} 
\label{th:convergenceFixpoint}
Let $C=\{c_1, \ldots, c_h\}$ be any set of constraints, let $F=\{f_{c_1}, \ldots, f_{c_h}\}$ be the functions computed by constraints in $C$, and let $(\Sigma, \preceq)$ be any partial ordering over $\Sigma$ satisfying the finite chain property. If functions in $F$ are inflationary on $\Sigma$ and $\{f\in F\,|\,f(\sigma)\neq\sigma\}\subseteq S\subseteq F$, then $\langle\rho, \sigma, D, S\rangle\trans_s \langle\sigma', D'\rangle$ and $\sigma'$ is a common fixpoint of the functions in $F$ such that $\sigma \preceq \sigma'$.
\end{theorem}

\begin{sketch} Consider the sequence  $\langle S_0, S_1, \ldots\rangle$ of scheduling sets resulting from a recursive application of rule {\sc Solver-2} terminated by rule {\sc Solver-1} (see Figure~\ref{fi:semantics}), with $S_0=S$. Let $c_i=\pick(S_i)$ the constraint executed at step $i$, and let $q=\langle\sigma_0, \sigma_1, \sigma_2,  \ldots\rangle$ be the function iteration such that $\sigma_0=\sigma$ and $\sigma_{i+1}=f_{c_i}(\sigma_i)$. We prove that $q$ is regular. Notice that $S_0=S$ contains initially all functions for which $\sigma_0$ is not a fixpoint. Furthermore, $S_{i+1}$ is obtained from $S_i$ by removing $c_i$ and adding at least all constraints for which $\sigma_{i+1}$ is not a fixpoint. It remains to show that all constraints are activated at some step, i.e., they are eventually removed from $S$. This can be proved by observing that an inflationary function $f_{c_i}$ either leaves the store unchanged, and therefore $|S|$ decreases by one, or produces a store $\sigma_{i+1}=f_{c_i}(\sigma_i)$ strictly larger than $\sigma_i$, i.e., $\sigma_i\preceq \sigma_{i+1}$ and $\sigma_i\neq\sigma_{i+1}$. By the finite chain property, this cannot happen indefinitely, so $S$ eventually gets empty. Since $q$ is regular, the proof follows from Lemma~\ref{le:fixpoint}.
\end{sketch}

\noindent Assuming that functions in Lemma~\ref{le:fixpoint} and Theorem~\ref{th:convergenceFixpoint} are also {\em monotonic}, it is possible to prove that the solver always converges to the {\em least} common fixpoint, yielding deterministic results independently of the scheduling order. We recall that a function $f$ is monotonic if $\sigma \preceq \sigma'$ implies $f(\sigma)\preceq f(\sigma')$ for all $\sigma, \sigma'\in\Sigma$.

\subsection{Acyclic Systems of Constraints} 
\label{ss:constraint-examples}

In this section we show that, if a system of constraints is acyclic, then our solver always converges deterministically to the correct result, without the need for programmers to prove any stabilization properties of their constraints. We notice that this is the most common case in many applications, and several efficient techniques can be adopted by constraint solvers to automatically detect cycles introduced by programming errors~\cite{Zanden01}. In particular, we prove termination of our solver on any system of constraints that models a computational circuit subject to incremental changes of its input. For the sake of simplicity, we focus on single-output constraints, to which any multi-output constraint can be reduced.

A circuit is a directed acyclic graph $G=(V,E)$ with values computed at the nodes, referred to as {\em gates}~\cite{AlpernHRSZ90}. Each node $u$ is associated with an output function $g_u$ that computes a value $val(u)=g_u(val(v_1),...,$ $val(v_{d_u}))$, where $d_u$ is the indegree of node $u$ and, for each $i\in[1,d_u]$, arc $(v_i,u)\in E$. Arcs entering $u$ are ordered and, if $d_u=0$, $u$ is called an input gate (in this case, $g_u$ is constant). The gate values and functions may have any data types. For simplicity, we will assume that there is only one gate in the graph with outdegree $0$: the value computed at this gate is the output of the circuit. 
The circuit value problem is to update the output of the circuit whenever the value of an input gate is changed. This problem is equivalent to the scenario where the gate values $val(u)$ are reactive memory cells, and each non-input gate $u$ is computed by a constraint $c_u$ that assigns $val(u)$ with $g_u(val(v_1),...,val(v_{d_u}))$. A circuit update operation changes the value $val(u)$ of any input gate $u$ to a new constant. Any such update triggers the solver with $S=\{c_{u'}\,|\,(u,u')\in E\}$. We now show that the solver updates correctly the circuit output.

Let $n$ be the number of circuit nodes, and let $u_1,u_2,...,u_n$ be any topological ordering of the nodes, where $u_n$ is the output gate of the circuit. Let  $\sigma\in \Sigma$ be a store with $dom(\sigma)=\{val(u_i)\,|\,i\in [1,n]\}$. We say that a value $val(u_i)$ is incorrect in $\sigma$ if $\sigma$ is not a fixpoint for $c_{u_i}$. Let $b(\sigma)=\sum_{i=1}^{n}2^i\cdot \chi_{\sigma}(u_i)$, where $\chi_{\sigma}(u_i)=1$ if value $val(u_i)$ is incorrect in $\sigma$, and $0$ otherwise. We define a partial ordering $(\Sigma, \preceq)$ as follows: 

\begin{definition}
\label{de:poset-circuit}
For any two stores $\sigma$ and $\sigma'$ in $\Sigma$, we say that $\sigma\preceq\sigma'$ if $b(\sigma)\ge b(\sigma')$.
\end{definition}

\noindent Relation $\preceq$ is clearly reflexive, antisymmetric, and transitive. Moreover, the constraints $c_{u}$ compute inflationary functions on $\Sigma$. Let $\sigma'=f_{c_u}(\sigma)$ be obtained by evaluating contraint $c_u$ in store $\sigma$. If $val(u)$ is correct in $\sigma$ (i.e., $\chi_{\sigma}(u)=0$), then $\sigma'=\sigma$. Otherwise, $\chi_{\sigma}(u)=1$, $\chi_{\sigma'}(u)=0$, and $\chi_{\sigma}(\hat{u})=\chi_{\sigma'}(\hat{u})$ for all gates $\hat{u}$ that precede $u$ in the topological ordering. This implies that $b(\sigma)\ge b(f_{c_u}(\sigma))$, and thus $\sigma\preceq f_{c_u}(\sigma)$.
Since $b(\cdot)$ can assume only $2^n$ possible values, $(\Sigma, \preceq)$ also satisfies the finite chain property.
After updating an input gate $u$, all constraints for which $\sigma$ is no longer a fixpoint are included in the set $S=\{c_{u'}\,|\,(u,u')\in E\}$ on which the solver is started. Hence, all the hypotheses of Theorem~\ref{th:convergenceFixpoint} hold and the solver converges to store $\sigma'\in\Sigma$ that is a common fixpoint of all $f_{c_{u}}$, i.e., a store $\sigma'$ in which all gate values are correct. This implies that the circuit output is also correct. It is not difficult to prove that, if we let function $\pick(S)$ return constraints in topological order, then all gates that may be affected by the input change are evaluated exactly once during the update.

\subsection{Cyclic Systems of Constraints: an Example} 
\label{ss:dynamic-sssp}
\label{ss:cyclic}

Differently from previous approaches to solving one-way dataflow constraints~\cite{DemersRT81,Hudson91, AlpernHRSZ90, Hoover87}, which were targeted to acyclic dependencies, our abstract machine can handle the most general case of cyclic constraints embedded within an imperative program. This opens up the possibility to address problems that would not be solvable using acyclic dataflow graphs, backed up with a formal machinery to help designers prove their convergence properties (Section~\ref{ss:termination}). We exemplify this concept by considering the well known problem of maintaining distances in a graph subject to local changes to its nodes or arcs. In the remainder of this section we show how to specify an incremental variant of the classical Bellman-Ford's single-source shortest path algorithm~\cite{b-orp-58} in terms of a (possibly cyclic) system of one-way constraints. Compared to purely imperative specifications~\cite{Demetrescu+2001a}, the formulation of the incremental algorithm in our mixed imperative/dataflow framework is surprisingly simple and requires just a few lines of code. By suitably defining a partial order on $\Sigma$ and an appropriate $\pick$ function, we show that our solver finds a correct solution within the best known worst-case time bounds for the problem.

\begin{figure}
\begin{boxedminipage}{7.3cm}
\begin{footnotesize}
\begin{tabbing}
{\tt insert}$(u,v,w)$:\\ 
~~~~~~~\=$E:=E\cup \{(u,v)\}$\\
\>$w(u,v)$ := $w$\\
\>\newcons$(\,\underbrace{\wif\:d[u] + w(u,v) < d[v]\:\wthen\:d[v]\:\mbox{:=}\:d[u] + w(u,v)}_{c_{uv}}\,)$
\end{tabbing}
\end{footnotesize}\hfill
\vspace{-6mm}\begin{footnotesize}
\begin{tabbing}
{\tt decrease}$(u,v,\delta)$:\\ 
~~~~~~~\=$w(u,v)$ := $w(u,v)-\delta$
\end{tabbing}
\end{footnotesize}
\end{boxedminipage}
\vspace{0mm}
\caption{Incremental shortest path updates in a mixed imperative/\dataflow\ style.}
\label{fi:incremental-sp}
\end{figure}

\paragraph{Incremental Shortest Paths.} Let $G=(V,E,w)$ be a directed graph with real edge weights $w(u,v)$, and let $s$ be a source node in $V$. We consider the incremental shortest path problem that consists of updating the distances $d[u]$ of all nodes $u\in V$ from the source $s$ after inserting any new edge in the graph, or decreasing the weight of any existing edge. For the sake of simplicity we assume that no negative-weight cycles are introduced by the update and that, if a node $u$ is unreachable from the source, its distance $d[u]$ is $+\infty$. 

\paragraph{Update Algorithm.} The incremental shortest path problem can be solved in our framework as follows. We keep edge weights and distances in reactive memory. Assuming to start from a graph with no edges, we initialize $d[s]$ := $0$ and $d[u]$ := $+\infty$ for all $u\neq s$. The pseudocode of update operations that insert a new edge and decrease the weight of an existing edge by a positive amount $\delta$ are shown in Figure~\ref{fi:incremental-sp}. 
Operation {\tt insert}$(u,v,w)$ adds edge $(u,v)$ to the graph with weight $w$ and creates a new constraint $c_{uv}$ for the edge: $c_{uv}$ simply {\em relaxes} the edge if Bellman's inequality $d[u]+w(u,v)\ge d[v]$ is violated~\cite{b-orp-58}. The constraint is immediately executed after creation (see rule {\sc NewCons-n1} in Figure~\ref{fi:semantics}) and the three pairs $(d[u], c_{uv})$, $(d[v], c_{uv})$, and $(w(u,v), c_{uv})$ are added to the set of dependencies $D$. Any later change to $d[u]$, $d[v]$ or $w(u,v)$, which may violate the inequality $d[u]+w(u,v)\ge d[v]$, will cause the re-execution of $c_{uv}$. Decreasing the weight of an existing edge $(u,v)$ by any positive constant $\delta$ with {\tt decrease}$(u,v,\delta)$ can be done by just updating $w(u,v)$. In view of rule {\sc Asgn-n2} of Figure~\ref{fi:semantics}, the system reacts to the change and re-executes automatically $c_{uv}$ and any other affected constraints.

Using the machinery developed in Section~\ref{ss:termination} and suitably defining a partial order on $\Sigma$ and an appropriate $\pick$ function, we now show that our solver finds a correct solution within the best known worst-case time bounds for the problem, i.e., it updates distances correctly after any {\tt insert} or {\tt decrease} operation. 

\paragraph{Termination and Correctness.} For the sake of convenience, we denote by $d_{\sigma}[u]=\sigma(d[u])$ and by $w_{\sigma}(u,v)=\sigma(w(u,v))$ the distance of node $u$ and the weight of edge $(u,v)$ in store $\sigma$, respectively. For any two graphs $G_1$ and $G_2$ on the same vertex set, we denote by $G_{1}\uplus G_{2}$ the multigraph with vertex set $V$ and edge set $E_{1}\uplus E_{1}$, where $\uplus$ indicates the join of multisets: the same edge may thus appear twice in $E_{1}\uplus E_{2}$, possibly with different weights. 
Let us focus our attention on a restricted set of stores, which encompasses all possible configurations of the  reactive memory during the execution of the solver triggered by an update:

\begin{definition}
Let $G_{old}=(V, E_{old}, w_{old})$ and $G=(V, E, w)$ be the graph before and after inserting a new edge or decreasing the weight of an edge, respectively. We denote by
$\Sigma_{sp}\subseteq \Sigma$ the set of all functions $\sigma:Loc\rightarrow Val$ such that:
\begin{itemize}
\item $dom(\sigma)=\{d[u]\,|\,u\in V\}\cup\{w(u,v)\,|\,(u,v)\in E\}\subseteq Loc$;
\item for each $u\in V$, $d_{\sigma}[u]$ is the weight of a simple path (i.e., with no repeated nodes) from $s$ to $u$ in $G_{old}\uplus G$;
\item for each $(u,v)\in E$, $w_{\sigma}(u,v)$ is fixed as the weight of edge $(u,v)$ in $G$.
\end{itemize}
\end{definition}

\noindent Notice that, as simple paths in a graph are finite, the number of possible values each $d_{\sigma}[u]$ can attain is finite, and therefore $\Sigma_{sp}$ is a finite set. We define a partial ordering $(\Sigma_{sp}, \preceq)$ on $\Sigma_{sp}$ as follows: 

\begin{definition}
\label{de:poset-shortest-paths}
Let $\sigma$ and $\sigma'$ be any two stores in $\Sigma_{sp}$. We say that $\sigma\preceq\sigma'$ if $d_{\sigma}[u]\ge d_{\sigma'}[u]$ for all $u\in V$.
\end{definition}

\noindent Relation $\preceq$ is reflexive, antisymmetric, and transitive. Moreover, since $\Sigma_{sp}$ is finite, $(\Sigma_{sp}, \preceq)$ satisfies the finite chain property. We now prove that constraints $c_{uv}$ compute inflationary functions on $\Sigma_{sp}$.

\begin{lemma}
Functions $f_{c_{uv}}$ computed by constraints $c_{uv}$ of Figure~\ref{fi:incremental-sp} are inflationary with respect to the partial ordering $(\Sigma_{sp}, \preceq)$ of Definition~\ref{de:poset-shortest-paths}.
\end{lemma}

\begin{proof}
Let $(u,v)$ be any edge in $E$ and let $\sigma$ be any store in $\Sigma_{sp}$. If $\sigma=f_{c_{uv}}(\sigma)$ then clearly $\sigma\preceq f_{c_{uv}}(\sigma)$. Consider the case $\sigma\neq f_{c_{uv}}(\sigma)$. Notice that $\sigma$ and $f_{c_{uv}}(\sigma)$ can only differ in the value of memory location $d[v]$. Since there are no negative-weight cycles and $d_{\sigma}[u]$ is the weight of a simple path in $G_{old}\uplus G$, then so is $d_{f_{c_{uv}}(\sigma)}[v]=d_{\sigma}[u]+w_{\sigma}(u,v)$. Furthermore, as $c_{uv}$ never increases distances, then $d_{f_{c_{uv}}(\sigma)}[v]\le d_{\sigma}[v]$ . Therefore, $\sigma\preceq f_{c_{uv}}(\sigma)$. This shows that $f_{c_{uv}}$ is inflationary.
\end{proof}

\noindent It is not difficult to see that, if all distances are correct before an update, then the solver is started on a store $\sigma\in\Sigma_{sp}$. As all the hypotheses of Theorem~\ref{th:convergenceFixpoint} hold, the solver converges to store $\sigma'\in\Sigma_{sp}$ that is a common fixpoint of all $f_{c_{uv}}$. Therefore: (1) $d_{\sigma'}[u]$'s are weights of simple paths in $G$, and (2) they satisfy all Bellman's inequalities. It is well known~\cite{AMO93} that node labels satisfying both properties (1) and (2) are in fact the correct distances in the graph.

\paragraph{Running Time.} If we let function $\pick(S)$ return the constraint $c_{uv}\in S$ with the largest variation of $d[u]$ due to the update, then we can adapt results from~\cite{Demetrescu+2001a} and show that each $c_{uv}$ is executed by the solver at most once per update. If $\pick$ uses a standard priority queue with $O(\log n)$ time per operation, then the solver updates distances incrementally in $O(m\log n)$ worst-case time, even in the presence of negative edge weights (but no negative cycles). This can be reduced to $O(m+n\log n)$ by just creating one constraint per node, rather that one constraint per edge, and letting it relax all outgoing edges. This matches the best known algorithmic bounds for the problem~\cite{Demetrescu+2001a}. We remark that, recomputing distances from scratch with the best known static algorithm would require $O(mn)$ time in the worst case if there are negative edge weights, and $O(m+n\log n)$ time otherwise. In Section~\ref{ss:incremental-experiments} we will analyze experimentally the performances of our constraint-based approach showing that in practice it can be orders of magnitude faster than recomputing from scratch, even when all weights are non-negative.

\section{Embodiment into C/C++}
\label{se:C}

In this section we show how to apply the concepts developed in Section~\ref{se:model} to the C and C++ languages, deriving an extension that we call \dc. \dc\ has exactly the same syntax as C/C++, but operations that read or modify objects have a different semantics. All other primitives, including creating/deleting constraints, allocating/deallocating reactive objects, and opening/closing atomic blocks, are provided as runtime library functions\footnote{A detailed documentation of the \dc\ application programming interface, including stricter library naming conventions and several additional features not covered in this paper, is available at the URL: {\tt http://www.dis.uniroma1.it/\textasciitilde demetres/dc/}} (see Figure~\ref{fi:dc.h}).

\begin{figure}
\begin{boxedminipage}{\columnwidth}
\begin{footnotesize}
\begin{verbatim}
typedef void (*cons_t)(void*);
\end{verbatim}
\vspace{-4mm}\begin{verbatim}
int newcons(cons_t cons, void* param);
void delcons(int cons_id);
void* rmalloc(size_t size);
void rfree(void* ptr);
void begin_at();
void end_at();
void arm_final(int cons_id, cons_t final);
void set_comp(int (*comp)(void*, void*));
\end{verbatim}
\end{footnotesize}
\end{boxedminipage}
\nocaptionrule\caption{Main functions of the \dc\ language extension.}
\label{fi:dc.h}
\end{figure}

\paragraph{Reactive Memory Allocation.} Similarly to other automatic change propagation approaches (e.g.,~\cite{Amulet97,DBLP:conf/pldi/AcarBLTT10}), in \dc\ all objects allocated statically or dynamically are non-reactive by default. Reactive locations are allocated dynamically using library functions {\tt rmalloc} and {\tt rfree}, which work just like {\tt malloc} and {\tt free}, but on a separate heap.

\paragraph{Opening and Closing Atomic Blocks.} Atomic blocks are supported in  \dc\ using two library functions {\tt begin\_at} and {\tt end\_at}. Calling {\tt begin\_at} opens an atomic block, which should be closed with a matching call to {\tt end\_at}. Nested atomic blocks are allowed, and are handled using a counter of nesting levels so that the solver is only resumed at the end of the outer block, processing any pending constraints that need to be first executed or brought up to date as a result of the block's execution.

\paragraph{Creating and Deleting Constraints.} For the sake of simplicity, in Section~\ref{se:model} constraints have been modeled as ordinary commands. \dc\ takes a more flexible approach: constraints are specified as closures formed by a function that carries out the computation and a user-defined parameter to be passed to the function. Different constraints may therefore share the same function code, but have different user-defined parameters. New constraint instances can be created by calling {\tt newcons}, which takes as parameters a pointer {\tt cons} to a function and a user-defined parameter {\tt param}. When invoked in non-atomic normal execution mode, {\tt newcons} executes immediately function {\tt cons} with parameter {\tt param}, and logs all dependencies between the created constraint and the reactive locations read during the execution. If a constraint is created inside an atomic block (or inside another constraint), its first evaluation is deferred until the end of the execution of the current block (or constraint). All subsequent re-executions of the constraint triggered by modifications of the reactive cells it depends on will be performed with the same value of {\tt param} specified at the creation time. {\tt newcons} returns a unique id for the created constraint, which can be passed to {\tt delcons} to dispose of it.

\paragraph{Reading and Modifying Objects.}

Reading and modifying objects in reactive memory can be done in \dc\ by evaluating ordinary C/C++ expressions. We remark that no syntax extensions or explicit macro/function invocations are required.

\paragraph{Customizing the Scheduler.}~Differently from other approaches~\cite{Amulet97}, \dc\ allows programmers to customize the execution order of scheduled constraints. While the default \pick\ function of \dc\ (which gives higher priority to least recently executed constraints) works just fine in practice for a large class of problems (see Section~\ref{se:experiments}), the ability to replace it can play an important role for some specific problems, as we have seen in the incremental shortest paths example of Section~\ref{ss:dynamic-sssp}. \dc\ provides a function {\tt set\_comp} that installs a user-defined comparator to determine the relative priority of two constraints. The comparator receives as arguments the user-defined parameters associated with the constraints to be compared.

\paragraph{Final Handlers.} 
An additional feature of \dc, built on top of the core constraint handling mechanisms described in Section~\ref{fi:semantics}, is the ability to perform some finalization operations only when the results of constraint evaluations are stable, i.e., when the solver has found a common fixpoint. For instance, a constraint computing the attribute of a widget in a graphic user interface may also update the screen by calling drawing primitives of the GUI toolkit: if a redrawing occurs at each constraint execution, this may cause unnecessary screen updates and flickering effects. Another usage example of this feature will be given in Section~\ref{ss:repair}.

\dc\ allows users to specify portions of code for a constraint to be executed as final actions just before resuming the underlying imperative program interrupted by the solver activation. This can be done by calling function {\tt arm\_final} during constraint solving: the operation schedules a {\em final handler} to be executed at the end of the current solving session. The function takes as parameters a constraint id and a pointer to a final handler, or {\tt NULL} to cancel a previous request. A final handler receives the same parameter as the constraint it is associated to, but no dependencies from reactive locations are logged during its execution. All final handlers are executed in normal execution mode as a whole inside an atomic block.

\begin{figure}
\begin{boxedminipage}{\columnwidth}
\begin{scriptsize}
\begin{verbatim}
struct robject {
    void* operator new(size_t size) { return rmalloc(size); }
    void operator delete(void* ptr) { rfree(ptr); }
};

static void con_h(void*), fin_h(void*);
class rcons {
    int id;
  public:
    virtual void cons() = 0;
    virtual void final() {}
    rcons() { id = -1;   }
   ~rcons() { disable(); }
    void enable()  { if (id == -1) id = newcons(con_h, this); }
    void disable() { if (id != -1) { delcons(id); id = -1; }  }
    void arm_final()   { if (id != -1) arm_final(id, fin_h);  }
    void unarm_final() { if (id != -1) arm_final(id, NULL);   }
};

void con_h(void* p) { ((rcons*)p)->cons();  }
void fin_h(void* p) { ((rcons*)p)->final(); }
\end{verbatim}
\end{scriptsize}
\end{boxedminipage}
\nocaptionrule\caption{C++ wrapping of DC primitives.}
\label{fi:examples:dcpp}
\end{figure}

\paragraph{C++ Wrapping of \dc\ Primitives.} The examples in the remainder of this paper are based on a simple C++ wrapping of the \dc\ primitives, shown in Figure~\ref{fi:examples:dcpp}. We abstract the concepts of reactive object and constraint using two classes: {\tt robject} and {\tt rcons}. The former is a base class for objects stored in reactive memory. This is achieved by overloading the {\tt new} and {\tt delete} operators in terms of the corresponding \dc\ primitives {\tt rmalloc} and {\tt rfree}, so that all member variables of the object are reactive. Class {\tt rcons} is a virtual base class for objects representing dataflow constraints. The class provides a pure virtual function called {\tt cons}, to be defined in subclasses, which provides the user code for a constraint. An additional empty {\tt final} function can be optionally overridden in subclasses to define the finalization code for a constraint. The class also provides functions {\tt enable} and {\tt disable} to activate/deactivate the constraint associated with the object, and functions {\tt arm\_final} and {\tt unarm\_final} to schedule/unschedule the execution of final handlers.

\section{Applications and Programming Examples}
\label{ss:examples}

In this section, we discuss how \dc\ can improve C/C++ programmability in three relevant application scenarios. To the best of our knowledge, these applications have not been explored before in the context of dataflow programming. All the code we show is real.

\subsection{Incremental Computation}
\label{ss:incremental}

\begin{figure}
\begin{boxedminipage}{\columnwidth}
\begin{scriptsize}
\begin{verbatim}
template<typename T> struct node : robject, rcons {
    enum op_t { SUM, PROD };
    T    val;
    op_t op;
    node *left, *right;
    node(T v):   val(v), left(NULL), right(NULL) { enable(); }
    node(op_t o): op(o), left(NULL), right(NULL) { enable(); }
    void cons() {
        if (left == NULL || right == NULL) return;
        switch (op) {
            case SUM:  val = left->val + right->val; break;
            case PROD: val = left->val * right->val; break;
        }
    }
};
\end{verbatim}
\end{scriptsize}
\end{boxedminipage}
\nocaptionrule\caption{Incremental evaluation of expression trees.}
\label{fi:examples:exptrees}
\end{figure}

In many applications, the input data is subject to continuous updates that need to be processed efficiently. For instance, in a networking scenario, routers must react quickly to link failures by updating routing tables in order to minimize communication delays. When the input is subject to small changes, a program may fix incrementally only the portion of the output affected by the update, without having to recompute the entire solution from scratch. For many problems, efficient {\em ad hoc} algorithms are known that can update the output asymptotically faster that recomputing from scratch, delivering in practice speedups of several orders of magnitude~\cite{DemersRT81,ae-174-DFI2005}. Such dynamic algorithms, however, are typically difficult to design and implement, even for problems that are easy to be solved from-scratch. A language-centric approach, which was extensively explored in both functional and imperative programming languages, consists of automatically turning a conventional static algorithm into an incremental one, by selectively recomputing the portions of a computation affected by an update of the input. This powerful technique, known as self-adjusting computation~\cite{DBLP:conf/pldi/AcarBLTT10}, provides a principled way of deriving efficient incremental code for several problems. We now show that dataflow constraints can provide an effective alternative for specifying incremental programs. Later in this section we discuss differences and similarities of the two approaches.

\paragraph{Example.} To put our approach into the perspective of previous work on self-adjusting computation, we revisit the problem of incremental re-evaluation of binary expression trees discussed in~\cite{HammerAC09}. This problem is a special case of the circuit evaluation described in Section~\ref{ss:constraint-examples}: input values are stored at the leaves and the value of each internal node is determined by applying a binary operator (e.g., sum or product) on the values of its children. The final result of the evaluation is stored at the root. We start from the conventional node structure that a programmer would use for a binary expression tree, containing the type of the operation computed at the node (only relevant for internal nodes), the node's value, and the pointers to the subtrees. Our \dc-based solution (see Figure~\ref{fi:examples:exptrees}) simply extends the node declaration by letting it inherit from classes {\tt robject} and {\tt rcons}, and by providing the code of a constraint that computes the value of the node in terms of the values stored at its children. Everything else is exactly what the programmer would have done anyway to build the input data structure. An expression tree can be constructed by just creating nodes and connecting them in the usual way:

\begin{footnotesize}
\begin{verbatim}
node<int> *root    = new node<int>(node<int>::SUM);
root->left         = new node<int>(10);
root->right        = new node<int>(node<int>::PROD);
root->right->left  = new node<int>(2);
root->right->right = new node<int>(6);
\end{verbatim}
\end{footnotesize}

\begin{figure*}
\centerline{~~\includegraphics[width=13cm]{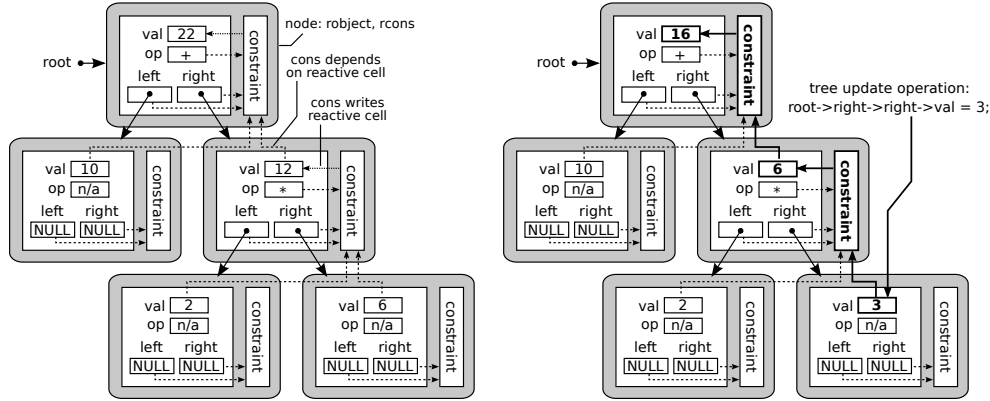}}
\bigskip
\caption{Reactive expression tree (left) and change propagation chain after a leaf value update (right).}
\label{fi:exptrees}
\end{figure*}

\noindent The example above creates the tree shown in Figure~\ref{fi:exptrees} (left). Since all fields of the node are reactive and each node is equipped with a constraint that computes its value, at any time during the tree construction, {\small\tt root->value} contains the correct result of the expression evaluation. We remark that this value not only is given for free without the need to compute it explicitly by traversing the tree, but is also updated automatically after any change of the tree. For instance, changing the value of the rightmost leaf with {\small\tt root->right->right->val = 3} triggers the propagation chain shown in Figure~\ref{fi:exptrees} (right). Other possible updates that would be automatically propagated include changing the operation type of a node or even adding/removing entire subtrees. Notice that a single change to a node may trigger the re-execution of the constraints attached to all its ancestors, so the total worst-case time per update is $O(h)$, where $h$ is the height of the tree. For a balanced expression tree, this is exponentially faster than recomputing from scratch. If a batch of changes are to be performed and only the final value of the tree is of interest, performance can be improved by grouping updates with {\small\tt begin\_at()} and {\small\tt end\_at()} so that the re-execution of constraints is deferred until the end of the batch, e.g.:

\begin{footnotesize}
\begin{verbatim}
begin_at();                // put the solver to sleep
root->op = node<int>::SUM; // change node operation type
delete root->right->left   // delete leaf
...                        // etc...
end_at();                  // wake up the solver
\end{verbatim}
\end{footnotesize}

\paragraph{Discussion.} \dc\ and imperative self-adjusting computation languages such as \ceal~\cite{DBLP:conf/pldi/AcarBLTT10} share the basic idea of change propagation, and reactive memory is very similar to \ceal's modifiables. However, the two approaches differ in a number of important aspects. In \ceal, the solution is initially computed by a core component and later updated by a {\em mutator}, which performs changes to the input. In \dc\ there is no explicit distinction between an initial run and a sequence of updates, and in particular there is no static algorithm that is automatically dynamized. Instead, programmers explicitly break down the solution to a complex problem into a collection of reactive code fragments that locally update small portions of the program state as a function of other portions.  This implies a paradigm shift that may be less straightforward for the average programmer than writing static algorithms, but it can make it easier to exploit specific properties of the problem at hand, which in some cases can be crucial for coding algorithms provably faster than recomputing from scratch.

Traceable data types~\cite{DBLP:conf/pldi/AcarBLTT10} have been recently introduced to extend the class of static algorithms that can be handled efficiently by self-adjusting computation: in a traceable data type, dependencies are tracked at the level of data structure operations, rather than individual memory locations, combining in a hybrid approach the benefits of automatic change propagation and those of {\em ad hoc} dynamic data structures. This can yield large asymptotic gains in dynamizing static algorithms that use basic abstract data types, such as dictionaries or priority queues. However, it is not clear that every conventional static algorithm can be effectively dynamized in this way: for some complex problems, it may be necessary to implement an {\em ad hoc} traceable data structure that performs all the incremental updates, thus missing the advantages of automatic change propagation and thwarting the automatic incrementalization nature of self-adjusting computation. In contrast, dataflow constraints  can explicitly deal with changes to the state of the program (when this is necessary to obtain asymptotic benefits), incorporating change-awareness directly within the code controlled by the change propagation algorithm without requiring traceable data structures.

\subsection{Implementing the Observer Design Pattern} 
\label{ss:observer-pattern}

As a second example, we show how the reactive nature of our framework can be naturally exploited to implement the observer software design pattern.
A common issue arising from partitioning a system into a collection of cooperating software modules is the need to maintain consistency between related objects. In general, a tight coupling of the involved software components is not desirable, as it would reduce their reusability. For example, graphical user interface toolkits almost invariably separate presentational aspects from the underlying application data management, allowing data processing and data presentation modules to be reused independently. The {\em observer software design pattern}~\cite{ObserverPOPL00} answers the above concerns by defining one-to-many dependencies between objects so that when one object (the \emph{subject}) changes state, all its dependents (the \emph{observers}) are automatically notified. A key aspect is that subjects send out notifications of their change of state, without having to know who their observers are, while any number of observers can be subscribed to receive these notifications (subjects and observers are therefore not tightly coupled). A widely deployed embodiment of this pattern is provided by the \qt\ application development framework~\cite{QtBook06}. 

\qt\ is based on a signal-slot communication mechanism: a signal is emitted when a particular event occurs, whereas a slot is a function that is called in response to a particular signal. An object acting as a subject emits signals in response to changes of its state by explicitly calling a special member function designated as a signal. Observers and subjects can be explicitly connected so that any signal emitted by a subject triggers the invocation of one or more observer slots. Programmers can connect as many signals as they want to a single slot, and a signal can be connected to as many slots as they need. Since the connection is set up externally after creating the objects, this approach allows objects to be unaware of the existence of each other, enhancing information encapsulation and reuse of software components. Subjects and observers can be created in \qt\ as instances of the {\tt QObject} base class. \qt's signal-slot infrastructure hinges upon an extension of the C++ language with three new keywords: {\tt signal} and {\tt slot}, to designate functions as signals or slots, and {\tt emit}, to generate signals. 

\paragraph{A Minimal Example: \qt\ vs. \dc.} To illustrate the concepts discussed above and compare \qt\ and \dc\ as tools for implementing the observer pattern, we consider a minimal example excerpted from the \qtfs\ reference documentation. The goal is to set up a program in which two counter variables {\tt a} and {\tt b} are connected together so that the value of {\tt b} is automatically kept consistent with the value of {\tt a}. The example starts with the simple declaration shown in Figure~\ref{fi:examples:counter}(a) (all except the framed box), which encapsulates the counter into an object with member functions {\tt value}/{\tt setValue} for accessing/modifying it. Figure~\ref{fi:examples:counter}(b) shows how the {\tt Counter} class can be modified in \qt\ so that counter modifications can be automatically propagated to other objects as prescribed by the observer pattern. First of all, the class inherits from \qt's {\tt QObject} base class and starts with the {\tt Q\_OBJECT} macro. Function {\tt setValue} is declared as a slot and it is augmented by calling explicitly the {\tt valueChanged} signal with the {\tt emit} keyword every time an actual change occurs. Since \qt\ {\tt Counter} objects contain both signal and slot functions they can act both as subjects and as observers. The following code snippet shows how two counters can be created and connected so that each change to the former triggers a change of the latter:

\begin{footnotesize}
\begin{verbatim}
Counter *a = new Counter, *b = new Counter;

QObject::connect(a, SIGNAL(valueChanged(int)),
                 b, SLOT(setValue(int)));

a->setValue(12); // a->value() == 12, b->value() == 12
b->setValue(48); // a->value() == 12, b->value() == 48
\end{verbatim}
\end{footnotesize}

\noindent The {\tt QObject::connect} call installs a connection between counters {\tt a} and {\tt b}: every time {\tt emit valueChanged(value)} is issued by {\tt a} with a given actual parameter, {\tt setValue(int value)} is automatically invoked on {\tt b} with the same parameter. Therefore, the call {\tt a->setValue(12)} has as a side-effect that the value of {\tt b} is also set to {\tt 12}. Conversely, the call {\tt b->setValue(48)} entails no change of {\tt a} as no connection exists from {\tt b} to {\tt a}. 

\begin{figure}
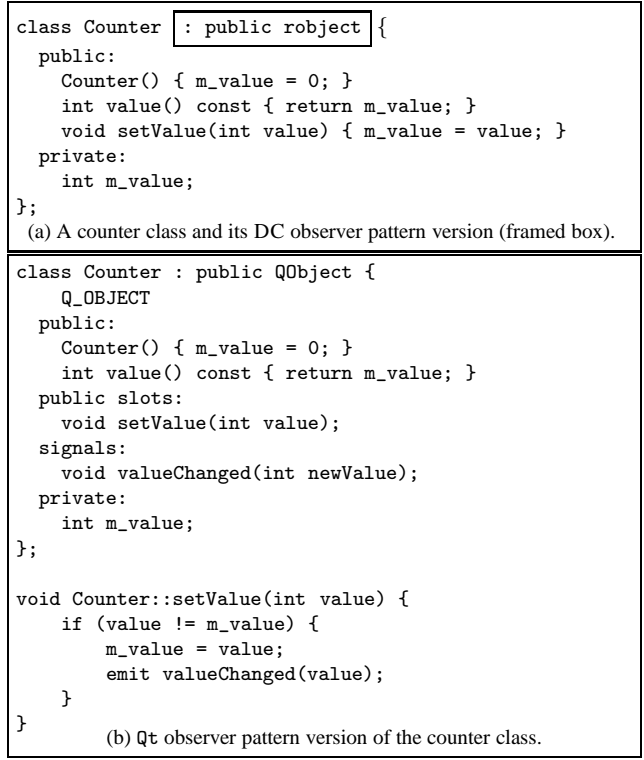

\begin{boxedminipage}{\columnwidth}
\begin{footnotesize}
\verb+class Counter +\fbox{\tt : public robject} {\tt \{} \\
\verb+  public:+\\
\verb+    Counter() { m_value = 0; }+\\
\verb+    int value() const { return m_value; }+\\
\verb+    void setValue(int value) { m_value = value; }+\\
\verb+  private:+\\
\verb+    int m_value;+\\
\verb+};+\\
\end{footnotesize}
\vspace{-1.1mm}\footnotesize{\centerline{(a) A counter class and its \dc\ observer pattern version (framed box). }}
\vspace{-1.5mm}
\end{boxedminipage}
\begin{boxedminipage}{\columnwidth}
\begin{footnotesize}
\begin{verbatim}
class Counter : public QObject {
    Q_OBJECT
  public:
    Counter() { m_value = 0; }
    int value() const { return m_value; }
  public slots:
    void setValue(int value);
  signals:
    void valueChanged(int newValue);
  private:
    int m_value;
};
 
void Counter::setValue(int value) {
    if (value != m_value) {
        m_value = value;
        emit valueChanged(value);
    }
}
\end{verbatim}
\end{footnotesize}

\vspace{-4mm} \footnotesize{\centerline{(b) \qt\ observer pattern version of the counter class.}}
\end{boxedminipage}
\nocaptionrule\caption{Observer pattern example excerpted from the \qtfs\ reference documentation: \dc\ vs. \qt\ implementation.}
\label{fi:examples:counter}
\end{figure}

The same result can be achieved in \dc\ by just letting the {\tt Counter} class of Figure~\ref{fi:examples:counter}(a) inherit from the {\tt robject} base class of Figure~\ref{fi:examples:dcpp}. As a result, the {\tt m\_value} member variable is stored in reactive memory. The prescribed connection between reactive counters can be enforced with a one-way dataflow constraint that simply assigns the value of {\tt b} equal to the value of {\tt a}:

\begin{footnotesize}
\begin{verbatim}
Counter *a = new Counter, *b = new Counter;

struct C : rcons {
    Counter *a, *b;
    C(Counter *a, Counter *b) : a(a), b(b) { enable(); }
    void cons() { b->setValue(a->value()); }
} c(a,b);

a->setValue(12); // a->value() == 12, b->value() == 12
b->setValue(48); // a->value() == 12, b->value() == 48
\end{verbatim}
\end{footnotesize}

\noindent We notice that the role of the {\tt QObject::connect} of the \qt\ implementation is now played by a dataflow constraint, yielding exactly the same program behavior.

\paragraph{Discussion.} The example above shows that \dc's runtime system handles automatically a number of aspects that would have to be set up explicitly by the programmers using \qt's mechanism:
\begin{itemize}
\item there is no need to define slots and signals, relieving programmers from the burden of extending the definition of subject and observer classes with extra machinery (see Figure~\ref{fi:examples:counter});
\item only actual changes of an object's state trigger propagation events, so programmers do not have to make explicit checks such as in {\tt Counter::setValue}'s definition to prevent infinite looping in the case of cyclic connections (see Figure~\ref{fi:examples:counter}(b));
\item \dc\ does not require extensions of the language, and thus the code does not have to be preprocessed before being compiled.
\end{itemize}

\noindent We sketch below further points that make dataflow constraints a flexible framework for supporting some aspects of component programming, putting it into the perspective of mainstream embodiments of the observer pattern such as \qt:

\begin{itemize}
\item in \dc, only subjects need to be reactive, while observers can be of any C++ class, even of third-party libraries distributed in binary code form. In \qt, third-party observers must be wrapped using classes equipped with slots that act as stubs;
\item relations between \qt\ objects are specified by creating explicitly one-to-one signal-slot connections one at a time; a single \dc\ constraint can enforce simultaneously any arbitrary set of many-to-many relations. Furthermore, as the input variables of a dataflow constraint are detected automatically, relations may change dynamically depending on the state of some objects;
\item \qt\ signal-slot connections let subjects communicate values to their observers; \dc\ constraints can compute the values received by the observers as an arbitrary function of the state of multiple subjects, encapsulating complex update semantics;
\item in \qt, an object's state change notification can be deferred by emitting a signal until after a series of state changes has been made, thereby avoiding needless intermediate updates; \dc\ programmers can control the granularity of change propagations by temporarily disabling constraints and/or by using {\tt begin\_at}/{\tt end\_at} primitives.
\end{itemize}

\begin{figure}
\begin{boxedminipage}{\columnwidth}
\begin{scriptsize}
\begin{verbatim}
 0 template<class T, class N> class snode : public rcons {
 1     map<N**, snode<T,N>*> *m;
 2     N     *head, **tail;
 3     snode *next;
 4     int    refc;
 5   public:
 6     snode(N *h, N **t, map<N**, snode<T,N>*> *m) :
 7         m(m), head(h), tail(t), next(NULL), refc(0) { 
 8         (*m)[tail] = this; 
 9         enable(); 
10     }
11    ~snode() { 
12         m->erase(tail); 
13         if (next != NULL && --next->refc == 0) delete next;
14     }
15     void cons() {
16         snode<T,N>* cur_next;
17
18         if (*tail != NULL) {
19             typename map<N**, snode<T,N>*>::iterator it =
20                 m->find( &(*tail)->next );
21             if (it != m->end()) 
22                  cur_next = it->second;
23             else cur_next = new snode<T,N>(*tail, 
24                                          &(*tail)->next, m);
25         } else cur_next = NULL;
26
27         if (next != cur_next) {
28             if (next != NULL && --next->refc == 0)
29                 next->arm_final();
30             if (cur_next != NULL && cur_next->refc++ == 0)
31                 cur_next->unarm_final();
32             next = cur_next;
33         }
34         if (head != NULL) T::watch(head);
35     }
36     void final() { delete this; }
37 };

38 template<class T, class N> class watcher {
39     snode<T,N> *gen;
40     map<N**, snode<T,N>*> m;
41   public:
42     watcher(N** h) { gen = new snode<T,N>(NULL, h, &m); }
43    ~watcher()      { delete gen; }
44 };
\end{verbatim}
\end{scriptsize}
\end{boxedminipage}
\nocaptionrule\caption{Data structure checking and repair: {\em list watcher}.}
\label{fi:examples:watcher}
\end{figure}

\subsection{Data Structure Checking and Repair} 
\label{ss:repair}

Long-living applications inevitably experience various forms of damage, often due to bugs in the program, which could lead to system crashes or wrong computational results. The ability of a program to perform automatic consistency checks and self-healing operations can greatly improve reliability in software development. One of the most common causes of faults is connected with different kinds of data structure corruptions, which can be mitigated using data structure repair techniques~\cite{Demsky:2003:ADR:949305.949314}. 

In this section, we show how dataflow constraints can be used to check and repair reactive data structures. We exemplify this concept by considering the simple problem of repairing a corrupt doubly-linked list~\cite{Malik:2009:CAD:1747491.1747568}. We first show how to build a generic {\em list watcher}, which is able to detect any changes to a list and perform actions when modifications occur. This provides an advanced example of \dc\ programming, where constraints are created and destroyed by other constraints. Differently from the expression trees of Section~\ref{ss:incremental}, where constraints are attributes of nodes, the main challenge here is how to let the watched list be completely unaware of the watcher, while still maintaining automatically a constraint for each node. The complete code of the watcher is shown in Figure~\ref{fi:examples:watcher}. The only assumpion our watcher makes on list nodes to be monitored (of generic type {\small\tt N}) is that they are reactive and contain a {\small\tt next} field pointing to the successor. The main idea is to maintain a {\em shadow list} of constraints that mirrors the watched list (Figure~\ref{fi:watcher}).  Shadow nodes are {\small\tt snode} objects containing pointers to the monitored nodes ({\small\tt head}) and to their next fields ({\small\tt tail}). A special generator shadow node ({\small\tt gen}) is associated to the reactive variable ({\small\tt list}) holding the pointer to the first node of the input list. A lookup table ({\small\tt m}) maintains a mapping from list nodes to the corresponding shadow nodes. The heart of the watcher is the constraint associated with shadow nodes (lines 15--35). It first checks if the successor of the monitored node, if any, is already mapped to a shadow node (lines 18--21). If not, it creates a new shadow node (line 23). Lines 27--33 handle the case where the successor of the shadow node has changed and its {\small\tt next} field has to be updated. Line 34 calls a user-defined {\small\tt watch} function (provided by template parameter {\small\tt T}), which performs any desired checks and repairs for an input list node. To dispose of shadow nodes when the corresponding nodes are disconnected from the list, we use a simple reference counting technique, deferring to a final handler the task of deallocating dead shadow nodes (line 36). 

The following code snippet shows how to create a simple repairer for a doubly-linked list based on the watcher of Figure~\ref{fi:examples:watcher}:

\begin{footnotesize}
\begin{verbatim}
struct node : robject { int val; node *next, *prev; };

struct myrepairer {
    static void watch(node* x) {
        // check
        if (x->next != NULL && x != x->next->prev) 
            // repair
            x->next->prev = x;                     
    }
};

// create reactive list head and repairer
node** list = ...;          
watcher<myrepairer,node> rep(list); 
// manipulate the list
...                                 
\end{verbatim}
\end{footnotesize}

\noindent The repairer object {\small\tt rep} checks if the invariant  property {\small\tt x == x->next->prev} is satisfied for all nodes in the list, and recovers it to a consistent state if any violation is detected during the execution of the program. We notice that several different watchers may be created to monitor the same list.

\begin{figure}
\centerline{\includegraphics[width=8.3cm]{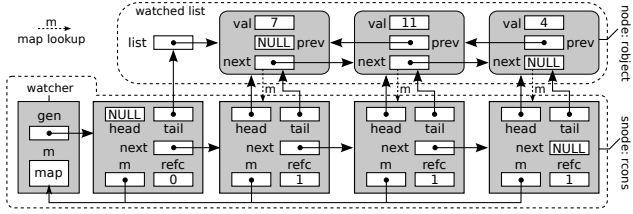}}
\bigskip
\caption{A reactive doubly-linked list, monitored by a {\em watcher}.}
\label{fi:watcher}
\end{figure}

\section{Implementation}
\label{se:implementation}

In this section we discuss how \dc\ can be implemented via a combination of runtime libraries, hardware/operating system support, and dynamic code patching, without requiring any source code preprocessing. The overall architecture of our \dc\ implementation, which was developed on a Linux IA-32 platform, is shown in Figure~\ref{fi:architecture}. At a very high level, the \dc\ runtime library is stratified into two modules: 1) a {\em reactive memory manager}, which defines the {\tt rmalloc} and {\tt rfree} primitives and provides support for tracing accesses to reactive memory locations; 2) a {\em constraint solver}, which schedules and dispatches the execution of constraints, keeping track of dependencies between reactive memory locations and constraints. We start our description by discussing how to support reactive memory, which is the backbone of the whole architecture.

\begin{figure}
\centerline{\includegraphics[width=8.4cm]{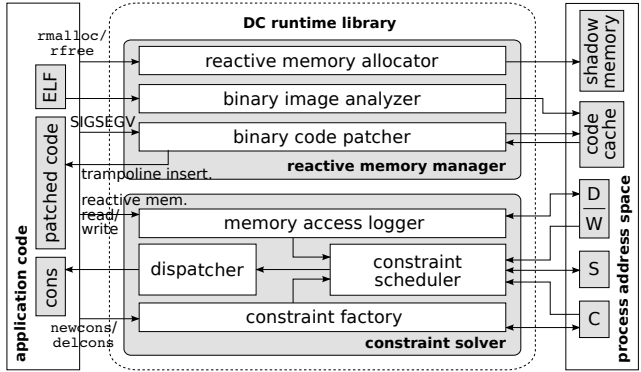}}
\bigskip
\caption{\dc's software architecture.}
\label{fi:architecture}
\end{figure}

\subsection{Reactive Memory}
\label{ss:reactive-implem}

Taking inspiration from transactional memories~\cite{1504203}, we implemented reactive memory using off-the-shelf memory protection hardware. Our key technique uses access violations (AV) combined with dynamic binary code patching as a basic mechanism to trace read/write operations to reactive memory locations.

\paragraph{Access Violations and Dynamic Code Patching.} Reactive memory is kept in a {\em protected} region of the address space so that any read/write access to a reactive object raises an AV. Since access violation handling is very inefficient, we use it just to detect incrementally instructions that access reactive memory. When an instruction $x$ first tries to access a reactive location, a segmentation fault with offending instruction $x$ is raised. In the SIGSEGV handler, we patch the trace $t$ containing $x$ by overwriting its initial 5 bytes with a jump to a dynamically recompiled trace $t'$ derived from $t$, which is placed in a code cache. In trace $t'$, $x$ is instrumented with additional inline code that accesses reactive locations without generating AVs, and possibly activates the constraint solver. Trace $t'$ ends with a jump that leads back to the end of $t$ so that control flow can continue normally in the original code. Since $t'$ may contain several memory access instructions, it is re-generated every time a new instruction that accesses reactive memory is discovered. To identify traces in the code, we analyze statically the binary code when it is loaded 
\begin{wrapfigure}{r}{1cm}
\hspace{-4.1mm}\includegraphics[width=1.3cm]{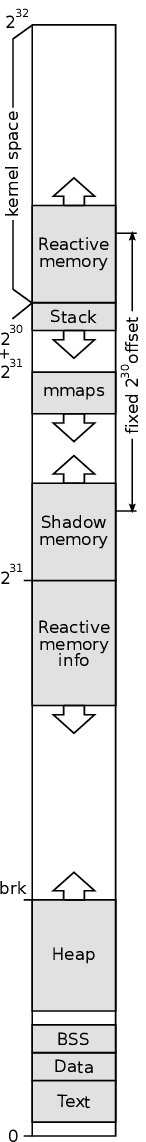}
\vspace{-3mm}
\end{wrapfigure} 
and we construct a lookup table that maps the address of each memory access instruction to the trace containing it. To handle the cases where a trace in a function $f$ is smaller than 5 bytes and thus cannot be patched, we overwrite the beginning of $f$ with a jump to a new version $f'$ of $f$ where traces are padded with trailing {\small\tt nop} instructions so that the smallest trace is at least 5-bytes long.

\paragraph{Shadow Memory and Address Redirecting.} To avoid expensive un-protect and re-protect page operations at each access to reactive memory, we mirror reactive memory pages with unprotected {\em shadow pages} that contain the actual data. The shadow memory region is kept under the control of our reactive memory allocator, which maps it onto physical frames with the {\tt mmap} system call. Any access to a reactive object is transparently redirected to the corresponding object in the shadow memory. As a result, memory locations at addresses within the reactive memory region are never actually read or written by the program. To avoid wasting memory without actually accessing it, reactive memory can be placed within the Kernel space, located in the upper 1GB of the address space on 32-bit Linux machines with the classical 3/1 virtual address split. Kernel space is flagged in the page tables as exclusive to privileged code (ring 2 or lower), thus an AV is triggered if a user-mode instruction tries to touch it. More recent 64-bit platforms offer even more flexibility to accomodate reactive memory in protected regions of the address space. We let the reactive memory region start at address $2^{30}+2^{31}=$ 0xC000000 and grow upward as more space is needed (see the figure on the right). The shadow memory region starts at address $2^{31}=$ 0x8000000 and grows upward, eventually hitting the memory mapping segment used by Linux to keep dynamic libraries, anonymous mappings, etc. Any reactive object at address $x$ is mirrored by a shadow object at address $x-\delta$, where $\delta=2^{30}=$ 0x4000000 is a fixed offset. This makes address redirecting very efficient.

\subsection{Constraint Solver} 
Our implementation aggregates reactive locations in 4-byte {\em words} aligned at 32 bit boundaries. The solver is activated every time such a word is read in constraint execution mode, or its value is modified by a write operation. The main involved units are (see Figure~\ref{fi:architecture}): 

\begin{enumerate}
\item A {\em dispatcher} that executes constraints, maintaining a global {\em timestamp} that grows by one at each constraint execution. For each constraint, we keep the timestamp of its latest execution.

\item A {\em memory access logger} that maintains the set of dependencies $D$ and a list $W$ of all reactive memory words written by the execution of the current constraint $c_{self}$, along with their initial values
before the execution. To avoid logging information about the same word multiple times during the execution of a constraint, the logger stamps each word with the time of the latest constraint execution that accessed it. Information is logged only if the accessed word has a timestamp older than the current global timestamp, which can only happen once for any constraint execution. To represent $D$, the logger keeps for each word $v$ the address of the head node of a linked list containing the id's of constraints depending upon $v$. 

\item A {\em constraint scheduler} that maintains the set of scheduled constraints $S$. By default $S$ is a priority queue, where the priority of a constraint is given by the timestamp of its latest execution: the scheduler repeatedly picks and lets the dispatcher execute the constraint with the highest priority, until $S$ gets empty. Upon completion of a constraint's execution, words are scanned and removed from $W$: for each $v\in W$ whose value has changed since the beginning of the execution, the constraint id's in the list of nodes associated with $v$ are added to $S$, if not already there.
\end{enumerate}

\noindent Nodes of the linked lists that represent $D$ and data structures $S$ and $W$ are kept in contiguous chunks allocated with {\tt malloc}. To support direct lookup, timestamps and dependency list heads for reactive memory words are stored in a contiguous {\em reactive memory info} region that starts at address $2^{31}=$ 0x8000000 and grows downward, eventually hitting the heap's {\tt brk}. 

A critical aspect is how to clean up old dependencies in $D$ when a constraint is re-evaluated. To solve the problem efficiently in constant amortized time per list operation, we keep for each node its insertion time into the linked list. We say that a node is {\em stale} if its timestamp is older than the timestamp of the constraint it refers to, and {\em up to date} otherwise. Our solver uses a lazy approach and disposes of stale nodes only when the word they refer to is modified and the linked list is traversed to add constraints to $S$. To prevent the number of stale nodes from growing too large, we use an incremental garbage collection technique.

\section{Experimental Evaluation}
\label{se:experiments}

In this section we present an experimental analysis of the performances of \dc\ in a variety of different settings, showing that our implementation is effective in practice. 

\begin{table*}
\begin{footnotesize}
\begin{center}
\begin{tabular}{|c||c|c|c|c|c|c||c|c|c||c|c|c||c|c|c|}\hline
\multicolumn{1}{|c||}{} &
\multicolumn{6}{|c||}{From-scratch time} &
\multicolumn{3}{|c||}{Propagation time} &
\multicolumn{3}{|c||}{Mem peak usage} &
\multicolumn{3}{|c|}{DC statistics} \\
\multicolumn{1}{|c||}{} &
\multicolumn{6}{|c||}{(secs)} &
\multicolumn{3}{|c||}{(msecs)} &
\multicolumn{3}{|c||}{(Mbytes)} &
\multicolumn{3}{|c|}{} \\\hline
Benchmark & conv & dc & ceal & $\frac{\textrm{dc}}{\textrm{conv}}$ & $\frac{\textrm{ceal}}{\textrm{conv}}$ & $\frac{\textrm{ceal}}{\textrm{dc}}$ & dc & ceal & $\frac{\textrm{ceal}}{\textrm{dc}}$ & dc & ceal & $\frac{\textrm{ceal}}{\textrm{dc}}$ & $\textrm{avg cons} \atop \textrm{per update}$ & $\textrm{instr} \atop \textrm{time}$ & $\textrm{patched} \atop \textrm{instr}$ \\\hline\hline

\texttt{adder} & 0.10 & 1.44 & 1.40 & 14.40 & 14.00 & 0.97 & 0.68 & 85.80 & 126.17 & 211.54 & 232.87 & 1.10 & 1.5 & 0.030 & 26 \\\hline
\texttt{exptrees} & 0.14 & 1.02 & 1.07 & 7.28 & 7.64 & 1.04 & 4.11 & 5.46 & 1.32 & 143.30 & 225.32 & 1.57 & 15.6 & 0.028 & 72 \\\hline
\texttt{filter} & 0.19 & 2.08 & 1.11 & 10.94 & 5.84 & 0.53 & 0.63 & 2.49 & 3.95 & 265.78 & 189.47 & 0.71 & 0.5 & 0.032 & 39 \\\hline
\texttt{halver} & 0.20 & 2.08 & 1.33 & 10.40 & 6.65 & 0.63 & 0.61 & 3.95 & 6.47 & 269.10 & 218.22 & 0.81 & 0.5 & 0.030 & 38 \\\hline
\texttt{mapper} & 0.19 & 2.04 & 1.30 & 10.73 & 6.84 & 0.63 & 0.61 & 2.63 & 4.31 & 261.53 & 214.34 & 0.81 & 0.5 & 0.032 & 39 \\\hline
\texttt{merger} & 0.19 & 2.12 & 1.37 & 11.15 & 7.21 & 0.64 & 0.66 & 4.43 & 6.71 & 284.41 & 218.21 & 0.81 & 0.5 & 0.031 & 57 \\\hline
\texttt{msorter} & 0.91 & 5.18 & 3.91 & 5.69 & 4.29 & 0.75 & 5.55 & 15.91 & 2.86 & 689.59 & 820.14 & 1.18 & 37.6 & 0.031 & 75 \\\hline
\texttt{reverser} & 0.18 & 2.04 & 1.30 & 11.33 & 7.22 & 0.63 & 0.62 & 2.63 & 4.24 & 267.45 & 214.34 & 0.80 & 0.5 & 0.030 & 37 \\\hline
\texttt{splitter} & 0.18 & 2.27 & 1.31 & 12.61 & 7.27 & 0.57 & 1.54 & 3.92 & 2.54 & 344.60 & 222.34 & 0.64 & 1.5 & 0.031 & 56 \\\hline
\end{tabular}
\end{center}
\end{footnotesize}
\caption{Performance evaluation of \dc\ versus \ceal, for a common set of benchmarks. Input size is $n=1,000,000$ for all tests except \texttt{msorter}, for which $n=100,000$.}
\label{ta:dc-ceal-comparison}
\end{table*}

\begin{figure*}
{\centerline{
\begin{tabular}{ccc}
\includegraphics[width=6.0cm]{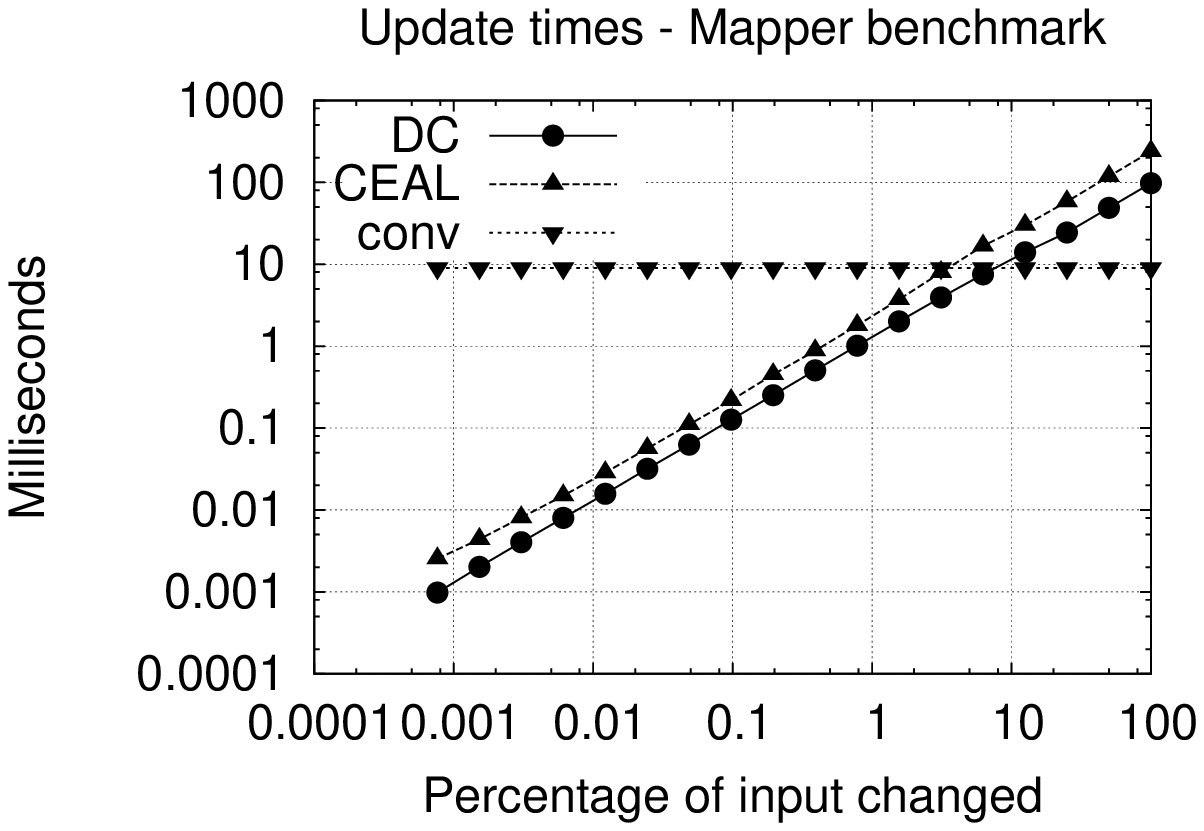} &
\includegraphics[width=6.0cm]{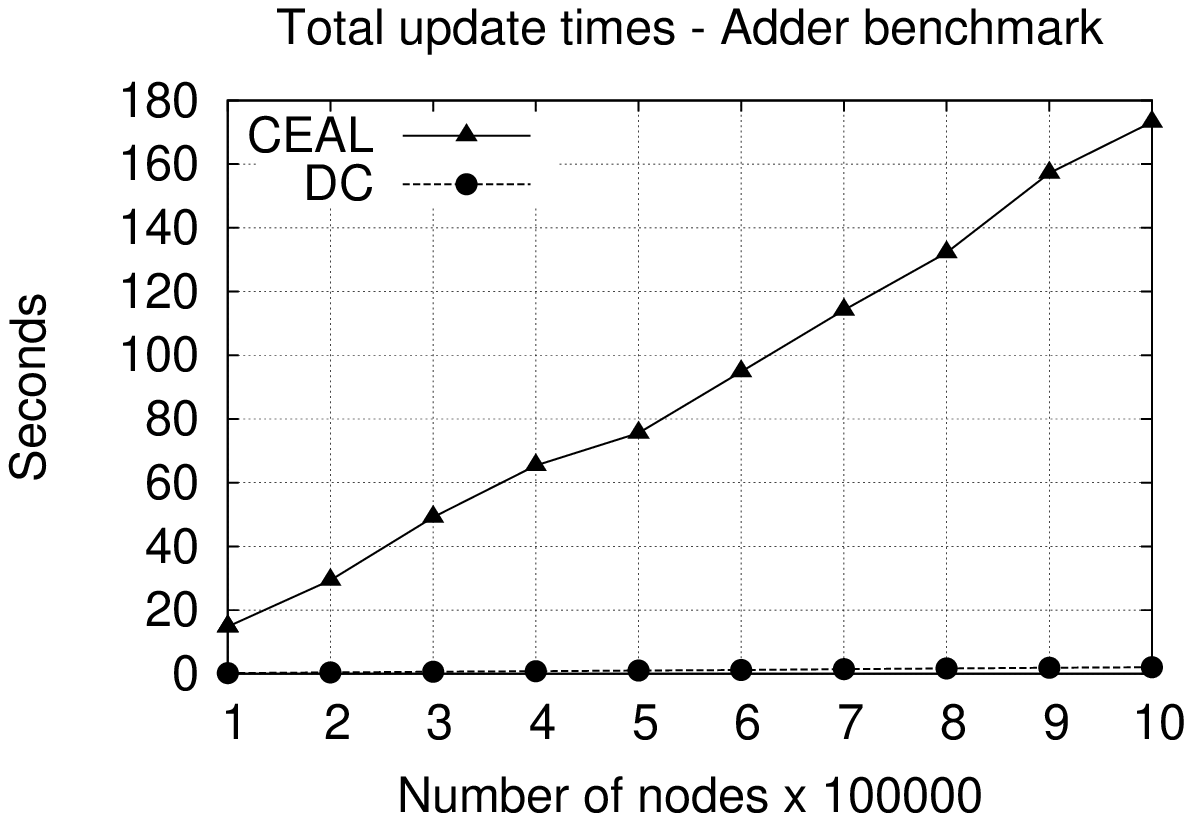} &
\includegraphics[width=6.0cm]{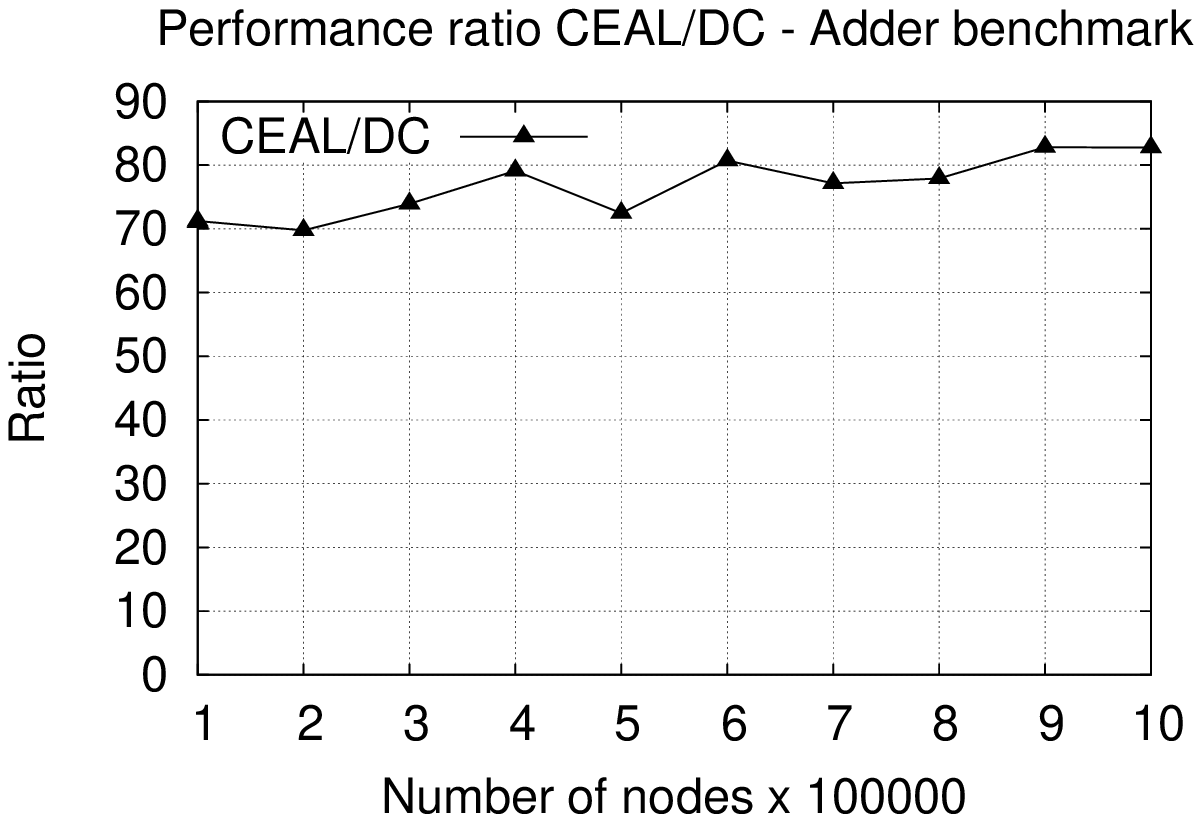}\\
(a) & (b) & (c)\\
\end{tabular}
}}
\smallskip
\caption{(a) Change propagation times on the {\tt mapper} benchmark for complex updates with input size $n=100,000$; (b-c) performance comparison of the  change propagation times of \dc\ and \ceal\ on the {\tt adder} benchmark.}
\label{fi:ceal}
\end{figure*}

\subsection{Benchmark Suite} 

We have evaluated \dc\ on a set of benchmarks that includes a variety of problems on lists, grids, trees, matrices, and graphs, as well as full and event-intensive interactive applications.

\begin{itemize}

\item {\em Linked Lists.} We considered several fundamental primitives on linear linked data structures, which provide a variety of data manipulation patterns. Our benchmarks include data structures for: computing the sum of the elements in a list ({\tt adder}), 
filtering the items of a list according to a given function ({\tt filter}), 
randomly assigning each element of a list to one of two output lists ({\tt halver}), 
mapping the items of a list onto new values according to a given mapping function ({\tt mapper}), 
merging two sorted lists into a single sorted output list ({\tt merger}),
producing a sorted version of an input list ({\tt msorter}), 
producing a reversed version of an input list ({\tt reverser}); 
splitting a list into two output lists, each containing only elements smaller or, respectively, greater than a given pivot ({\tt splitter}). All benchmarks are subject to operations that add or remove nodes from the input lists.

\item {\em Graphs and Trees.} Benchmarks in this class include classical algorithmic problems for routing in networks and tree computations:
\begin{itemize}
\item {\tt sp}: given a weighted directed graph and a source node $s$, computes the distances of all graph nodes from $s$. Graph edges are subject to edge weight decreases (see Section~\ref{ss:dynamic-sssp}).
\item {\tt exptrees}: computes the value of an expression tree subject to operations that change leaf values or operators computed by internal nodes (see Section~\ref{ss:incremental}).
\end{itemize}

\item {\em Linear Algebra.} We considered number-crunching problems on vectors and matrices, including the product of a vector and a matrix ({\tt vecmat}) and matrix multiplication ({\tt matmat}), subject to different kinds of updates of single cells as well as of entire rows or columns.

\item {\em Interactive Applications.} We considered both full real applications and synthetic worst-case scenarios, including:
\begin{itemize}
\item {\tt othello}: full application that implements the well-known board game in which two players in turn place colored pieces on a square board, with the goal of reversing as many of their opponent's pieces as possible;
\item {\tt buttongrid}: event-intensive graphic user interface application with a window containing $n\times n$ push buttons embedded in a grid layout. This is an extreme artificial scenario in which many events are generated, since a quadratic number of buttons need to be resized and repositioned to maintain the prescribed layout at each interactive resize event.
\end{itemize}

\end{itemize}

\noindent Some benchmarks, such as {\tt matmat} and {\tt sp}, are very computationally demanding. For all these benchmarks we have considered an implementation based on \dc, obtained by making the base data structures (e.g., the input list) reactive, and a conventional implementation in C based on non-reactive data structures. Interactive applications ({\tt othello} and {\tt buttongrid}) are written in the \qt-4 framework: change propagation throughout the GUI is implemented either using constraints (\dc\ versions), or using the standard signal-slot mechanism provided by \qt\ (conventional versions). To assess the performances of \dc\ against competitors that can quickly respond to input changes,  we have also considered highly tuned {\em ad-hoc} dynamic algorithms~\cite{Demetrescu+2001a, Ramalingam1996267} and incremental solutions realized in \ceal~\cite{HammerAC09}, a state-of-the-art C-based language for self-adjusting computation. Benchmarks in common with \ceal\ are {\tt adder},  {\tt exptrees},  {\tt filter}, {\tt halver}, {\tt mapper}, {\tt merger}, {\tt msorter}, {\tt reverser}, and {\tt splitter}. For these benchmarks, we have used the optimized implementations provided by Hammer {\em et al.}~\cite{HammerAC09}.

\subsection{Performance Metrics and Experimental Setup} 

We tested our benchmarks both on synthetic and on real test sets, considering a variety of performance metrics:

\begin{itemize}

\item {\em Running times:} we measured the time required to initialize the data structures with the input data (from-scratch execution), the time required by change propagation, and binary code instrumentation time. All reported times are wall-clock times, averaged over three independent trials. Times were measured with {\tt gettimeofday()}, turning off any other processes running in the background. 

\item {\em Memory usage:} we computed the memory peak usage as well as a detailed breakdown to assess which components of our implementation take up most memory (constraints, shadow memory, reactive memory, stale and non-stale dependencies, etc.).

\item {\em DC-related statistics:} we collected detailed profiling information including counts of patched instructions, stale dependencies cleanups, allocated/deallocated reactive blocks, created/deleted constraints, constraints executed per update, and distinct constraints executed per update.

\end{itemize}

\noindent All \dc\ programs considered in this section, except for {\tt sp} that will be discussed separetely, use the default timestamp-based comparator for constraint scheduling. 

\paragraph{Experimental Platform.} The experiments were performed on a PC equipped with a 2.10 GHz Intel Core 2 Duo with 3 GB of RAM, running Linux Mandriva 2010.1 with Qt 4.6. All programs were compiled with {\tt gcc} 4.4.3 and optimization flag {\tt -O3}.  

\begin{table*}
\begin{footnotesize}
\begin{center}
\begin{tabular}{|c|c|c||c||c|c||c|c||c|c|c||c|c|}\hline
\multicolumn{3}{|c||}{Road network} &
\multicolumn{1}{|c||}{From-scratch} &
\multicolumn{2}{|c||}{Propagation} &
\multicolumn{2}{|c||}{Speedup} &
\multicolumn{3}{|c||}{Mem peak usage} &
\multicolumn{2}{|c|}{Statistics} \\
\multicolumn{3}{|c||}{} &
\multicolumn{1}{|c||}{time (msec)} &
\multicolumn{2}{|c||}{time (msec)} &
\multicolumn{2}{|c||}{} &
\multicolumn{3}{|c||}{(Mbytes)} &
\multicolumn{2}{|c|}{}\\\hline
{Graph} & $n\cdot 10^3$ & $m\cdot 10^3$ & {sq} & {sp} & {rr} & $\frac{\textrm{sq}}{\textrm{sp}}$ & $\frac{\textrm{sq}}{\textrm{rr}}$ & {sp} & {rr} & {sq} & $\textrm{sp cons} \atop \textrm{per update}$ & $\textrm{rr node scans} \atop \textrm{per update}$\\
\hline\hline
NY & $264$ & $733$ & $50.99$ & $0.16$ & $0.07$ & $318.6$ & $728.4$ & $76.75$ & $26.62$ & $26.19$ & $143.9$ & $143.9$\\
\hline
BAY & $321$ & $800$ & $59.99$ & $0.15$ & $0.07$ & $399.9$ & $857.0$ & $84.84$ & $30.21$ & $29.82$ & $170.6$ & $170.5$\\
\hline
COL & $435$ & $1,057$ & $79.98$ & $0.28$ & $0.17$ & $285.6$ & $470.4$ & $108.61$ & $39.09$ & $38.97$ & $378.3$ & $378.2$\\
\hline
FLA & $1,070$ & $2,712$ & $192.97$ & $0.63$ & $0.35$ & $306.3$ & $551.3$ & $251.26$ & $93.42$ & $93.29$ & $687.5$ & $687.3$\\
\hline
NW & $1,207$ & $2,840$ & $236.96$ & $0.87$ & $0.54$ & $272.3$ & $438.8$ & $270.66$ & $102.15$ & $101.53$ & $1002.4$ & $1002.3$\\
\hline
NE & $1,524$ & $3,897$ & $354.94$ & $0.27$ & $0.16$ & $1314.5$ & $2218.3$ & $350.86$ & $132.85$ & $132.15$ & $320.2$ & $320.1$\\
\hline
\end{tabular}
\end{center}
\end{footnotesize}
\caption{Performance evaluation of \dc\ for incremental routing in US road
networks using up to 1.5 million constraints.}
\label{ta:sp-experiments}
\end{table*}

\begin{figure}
\centerline{
\includegraphics[width=8.8cm]{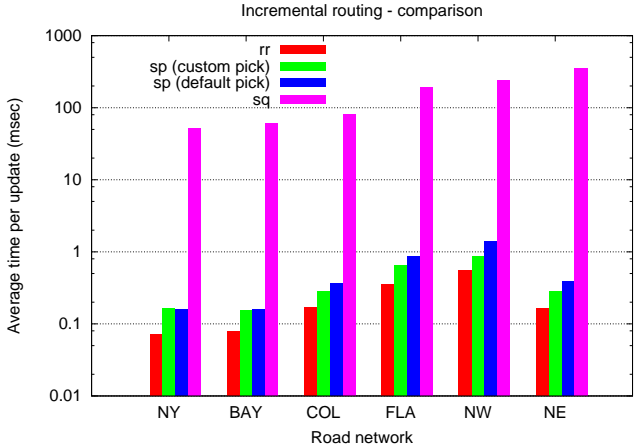}
}
\bigskip
\caption{Analysis of different {\tt pick} function definitions on the incremental routing problem.}
\label{fi:sssp_comparison}
\end{figure}

\subsection{Incremental Computation}
\label{ss:incremental-experiments}

As observed in Section~\ref{ss:dynamic-sssp}, the reactive nature of our mixed imperative/dataflow framework makes it a natural ground for incremental computation. In this section, we present experimental evidence that a constraint-based solution in our framework can respond to input updates very efficiently. We first show that the propagation times are comparable to state of the art automatic change propagation frameworks, such as \ceal~\cite{HammerAC09}, and for some problems can be orders of magnitude faster than recomputing from scratch. We then consider a routing problem on real road networks, and compare our \dc-based solution both to a conventional implementation and to a highly optimized \emph{ad hoc} dynamic algorithm supporting specific update operations.

\paragraph{Comparison to \ceal.} Table~\ref{ta:dc-ceal-comparison} summarizes the outcomes of our experimental comparison with the conventional version and with \ceal\ for all common benchmarks. Input size is $n=1,000,000$ for all tests (with the exception of {\tt msorter}, for which $n=100,000$), where $n$ is the length of the input list for the list-based benchmarks, and the number of nodes in the (balanced) input tree for {\tt exptrees}. Table~\ref{ta:dc-ceal-comparison}  reports  from-scratch execution times of both \dc\ and \ceal\ (compared to the corresponding conventional implementations), average propagation times in response to small changes of the input, memory usage and some \dc\ stats (average number of executed constraints per update, executable instrumentation time, and total number of patched instructions). The experiments show that our \dc\ implementation performs remarkably well. From-scratch times are on average a factor of $1.4$ higher than those of \ceal, while propagation times are smaller by a factor of 4 on average for all tests considered except the {\tt adder}, yielding large speed-ups over complete recalculation. In the case of the {\tt adder} benchmark, \dc\ leads by a huge margin in terms of propagation time (see Figure~\ref{fi:ceal}a and Figure~\ref{fi:ceal}b), which can be attributed to the different asymptotic performance of the algorithms handling the change propagation (constant for \dc, and logarithmic in the input size for the list reduction approach used by \ceal). We remark that the logarithmic bound of self-adjusting computation could be reduced to constant by using a traceable accumulator, as observed in Section~\ref{ss:incremental} (however, support for traceable data structures is not yet integrated in \ceal).

We also investigated how \dc\ and \ceal\ scale in the case of batches of updates that change multiple input items simultaneously. The results are reported in Figure~\ref{fi:ceal}a for the representative {\tt mapper} benchmark, showing that the selective recalculations performed by \dc\ and \ceal\ are faster than recomputing from scratch for changes up to significant percentages of the input.

\begin{figure*}
{\centerline{
\begin{tabular}{cccc}
\raisebox{23mm}{(a)} & \includegraphics[width=6.5cm]{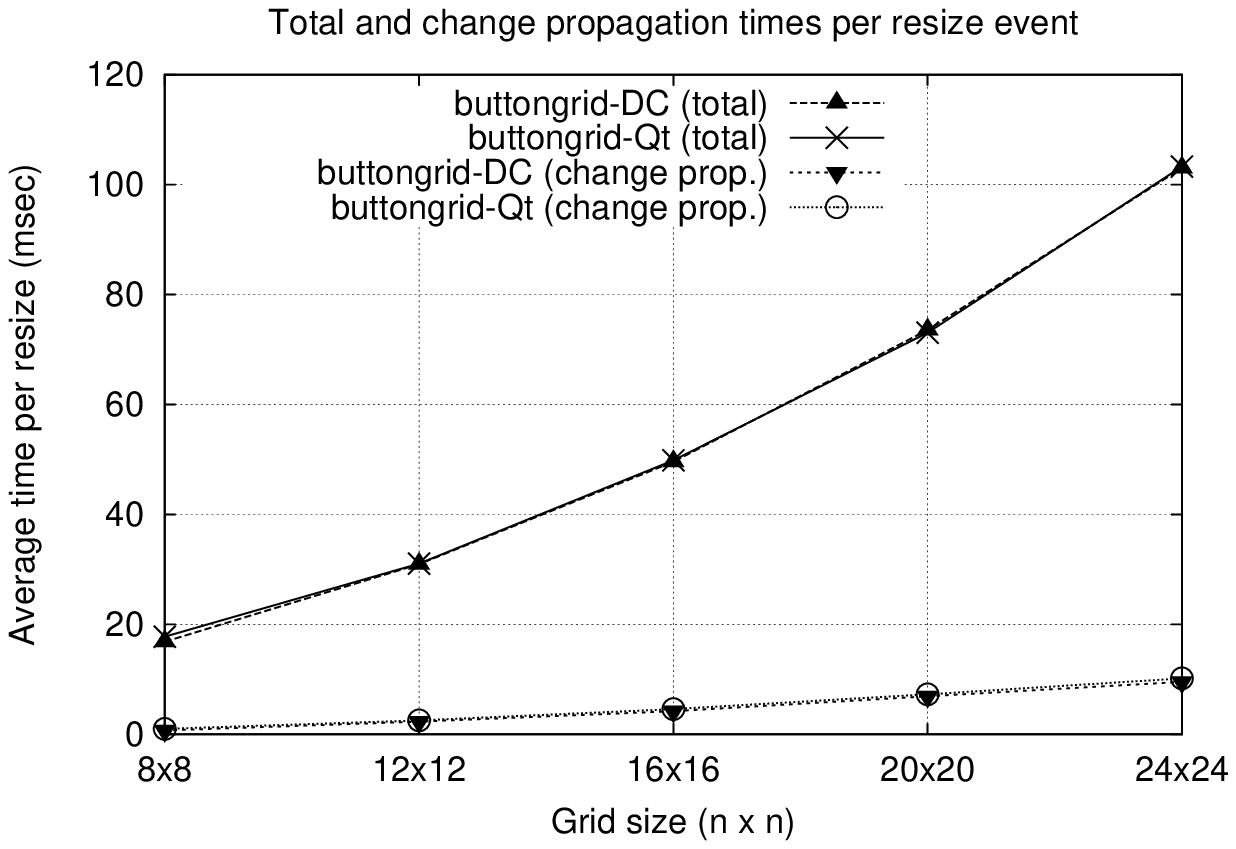}\hspace{5mm} &
\raisebox{23mm}{(b)} & \includegraphics[width=6.5cm]{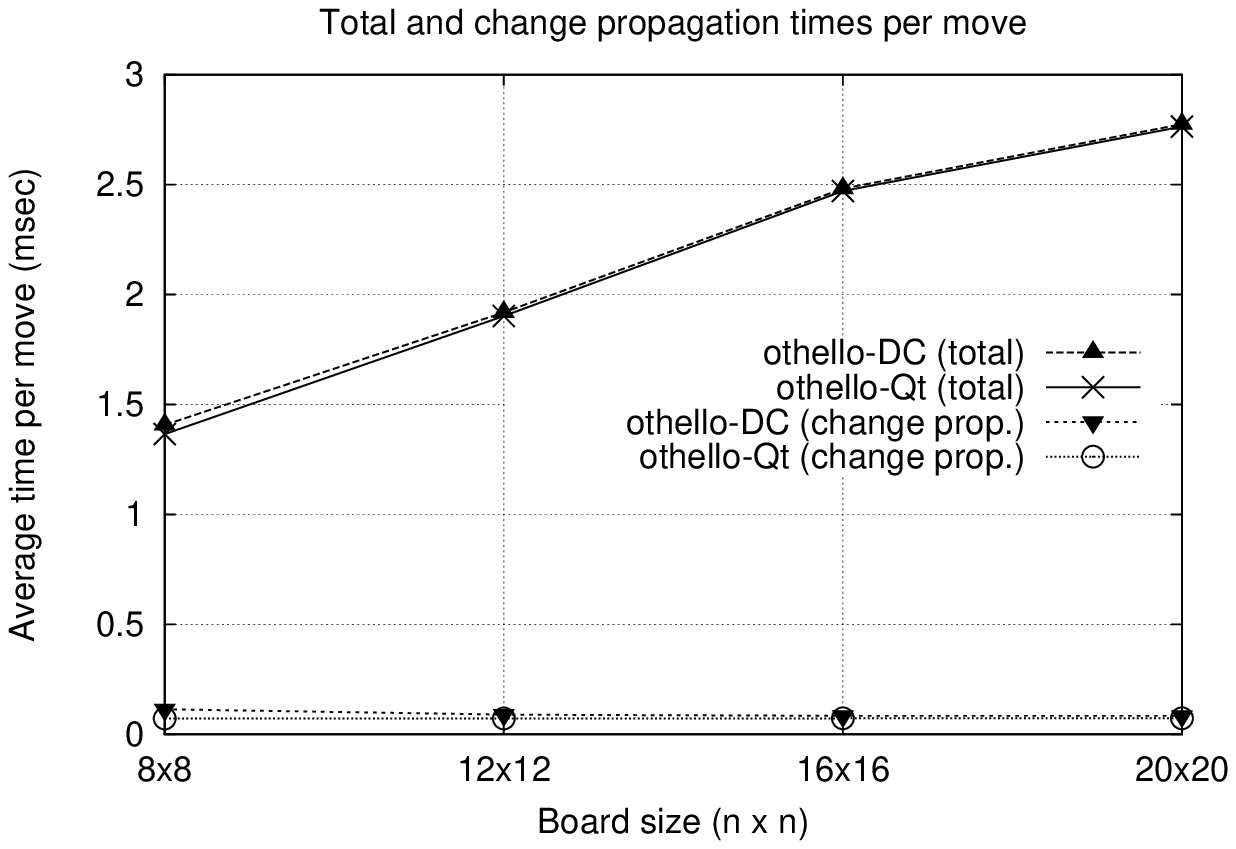}
\end{tabular}
}}
\smallskip
\caption{Comparison with signal-slot mechanism in \qt: (a) {\tt buttongrid}; (b) {\tt othello}.}
\label{fi:interactive}
\end{figure*}

\paragraph{Comparison to ad hoc Incremental Shortest Paths.} We now consider an application of the shortest path algorithm discussed in Section~\ref{ss:dynamic-sssp} to incremental routing in road networks. We assess the empirical performance of a constraint-based solution implemented in \dc\ ({\tt sp}) by comparing it with Goldberg's {\em smart queue} implementation of Dijkstra's algorithm ({\tt sq}), a highly-optimized C++ code used as the reference benchmark in the 9th DIMACS Implementation Challenge~\cite{ae-327-DGJ2009}, and with an engineered version of the \emph{ad hoc} incremental algorithm by Ramalingam and Reps ({\tt rr})~\cite{Demetrescu+2001a, Ramalingam1996267}.  Our code supports update operations following the high-level description given in Figure~\ref{fi:incremental-sp}, except that we create one constraint per node, rather than one constraint per edge. We used as input data a suite of US road networks of size up to 1.5 million nodes and 3.8 million edges derived from the UA Census 2000 TIGER/Line Files~\cite{Tiger02}. Edge weights are large and represent integer positive travel times. We performed on each graph a sequence of $m/10$ random edge weight decreases, obtained by picking edges uniformly at random and reducing their weights by a factor of 2. Updates that did not change any distances were not counted. 

The results of our experiments are shown in Table~\ref{ta:sp-experiments} and Figure~\ref{fi:sssp_comparison}. Both {\tt sp} and {\tt rr} were initialized with distances computed using {\tt sq}, hence we report from-scratch time only for this algorithm. Due to the nature of the problem, the average number of node distances affected by an update is rather small and almost independent of the size of the graph. Analogously to the incremental algorithm of Ramalingam and Reps, the automatic change propagation strategy used by our solver takes full advantage of this strong locality, re-evaluating only affected constraints and delivering substantial speedups over static solutions in typical scenarios. Our \dc-based implementation yields propagation times that are, on average, a factor of $1.85$ higher than the conventional \emph{ad hoc} incremental algorithm, but it is less complex, requires fewer lines of code, is fully composable, and is able to respond seamlessly to multiple data changes, relieving the programmer from the task of implementing explicitly change propagation. We also tested {\tt sp} with different types of schedulers. By customizing the {\tt pick} function of the default priority queue scheduler (giving highest priority to nodes closest to the source), a noticeable performance improvement has been achieved (see Figure~\ref{fi:sssp_comparison}). We also tried a simple stack scheduler, which, however, incurred a slowdown of a factor of 4 over the default scheduler.

\subsection{Comparison to \qt's Signal-slot Mechanism} 

Maintaining relations between widgets in a graphic user interface is one of the most classical applications of dataflow constraints~\cite{Zanden01}. We assess the performance of \dc\ in event-intensive interactive applications by comparing the \dc\ implementations of {\tt buttongrid} and {\tt othello} with the conventional versions built atop \qt's signal-slot mechanism.

In {\tt buttongrid}, each constraint computes the size and position of a button in terms of the size and position of adjacent buttons. We considered user interaction sessions with continuous resizing, which induce intensive scheduling activity along several propagation chains in the acyclic dataflow graph. In {\tt othello}, constraints are attached to cells of the game board (stored in reactive memory) and maintain a mapping between the board and its graphical representation: in this way, the game logic can be completely unaware of the GUI backend, as prescribed by the observer pattern (see Section~\ref{ss:observer-pattern}). For both benchmarks, we experimented with different grid/board sizes. Figure~\ref{fi:interactive} plots the average time per resize event ({\tt buttongrid}) and per game move ({\tt othello}), measured over 3 independent runs. Both the total time and the change propagation time are reported. For all values of $n$, the performance differences of the \dc\ and \qt\ conventional implementations are negligible and the curves are almost overlapped. Furthermore, the time spent in change propagation is only a small fraction of the total time, showing that the overhead introduced by access violations handling, instrumentation, and scheduling in \dc\  can be largely amortized over the general cost of widget management and event propagation in \qt\ and in its underlying layers.

\begin{figure*}
{\centerline{
\begin{tabular}{cccc}
\raisebox{23mm}{(a)} & \includegraphics[width=6.8cm]{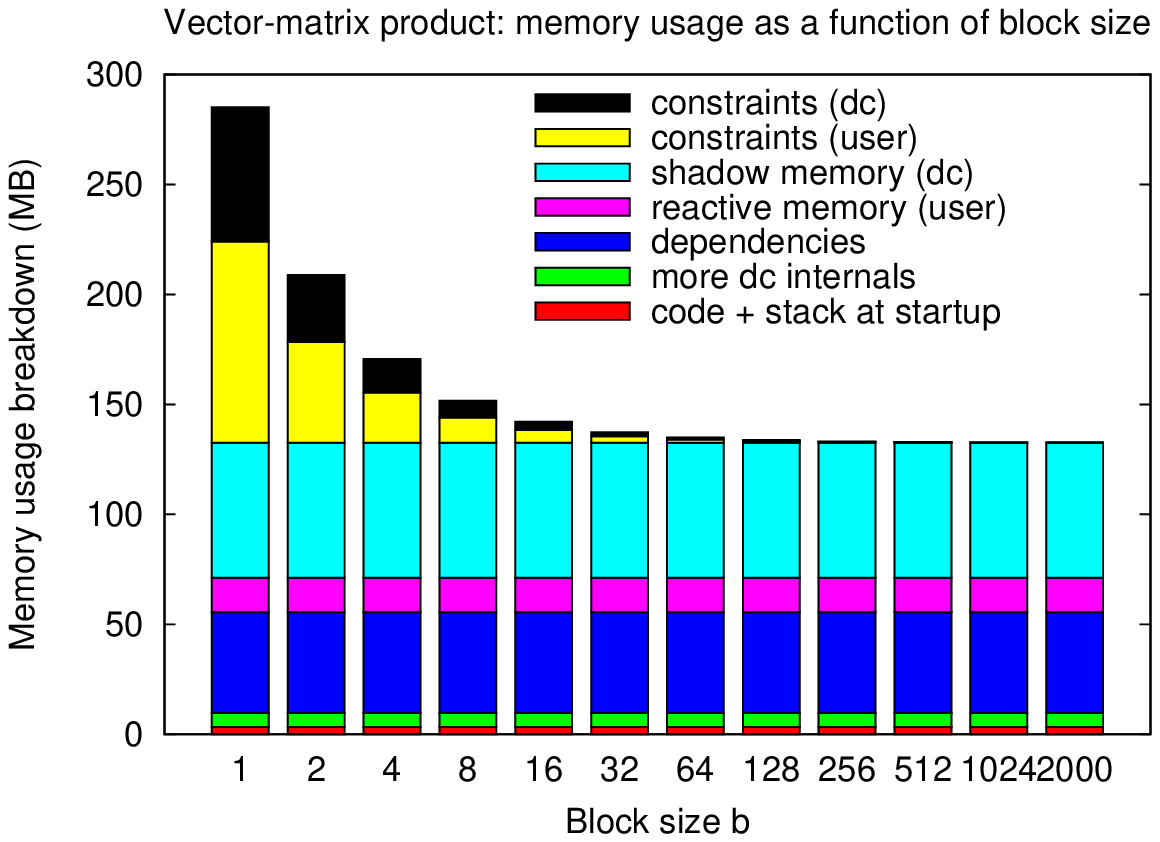}\hspace{10mm} &
\raisebox{23mm}{(b)} &\includegraphics[width=6.8cm]{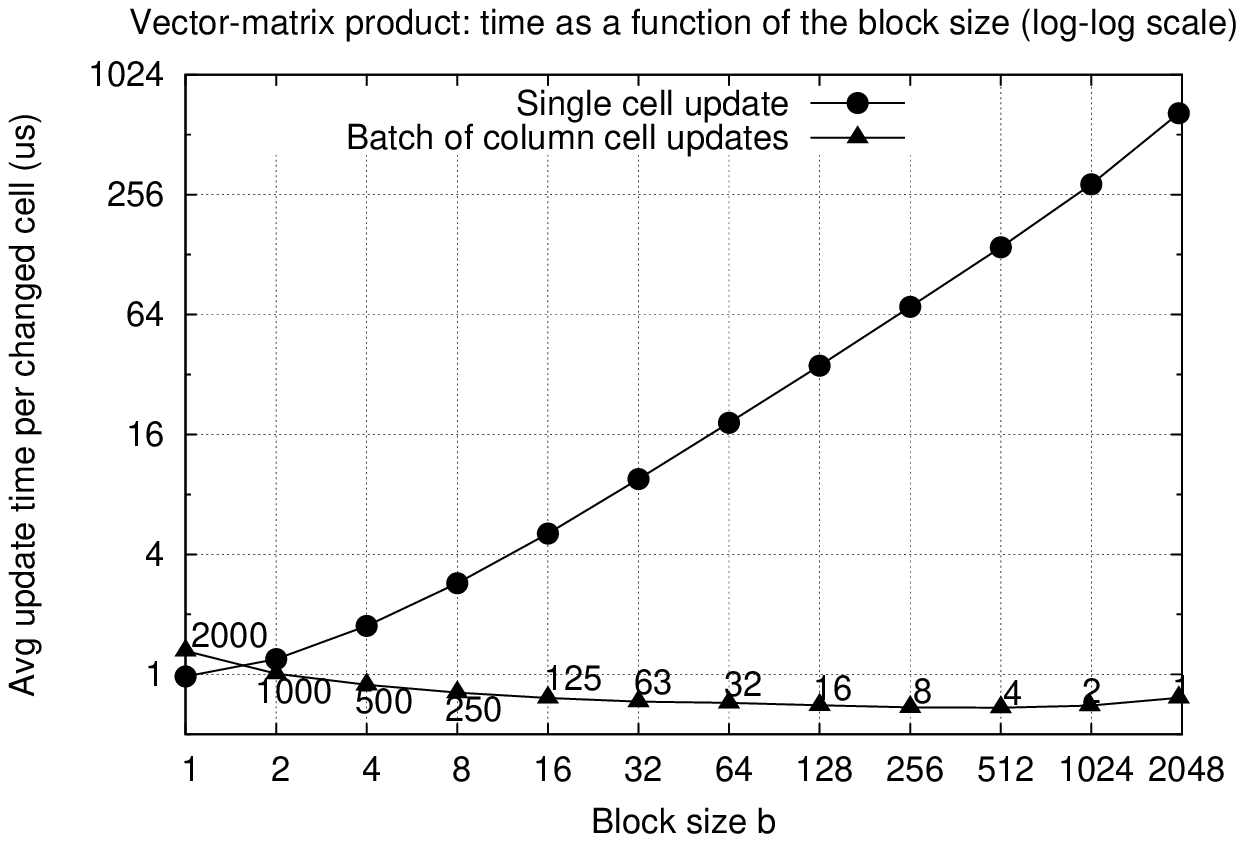}\\
\end{tabular}
}}
\caption{Analysis of {\tt vecmat} as a function of the block size: (a) memory usage; (b) update time for two different kinds of updates. Point labels indicate the number of constraints executed in each batch (for single cells updates, the number of executed constraints is always 1, but a constraint does more work as $b$ increases).}
\label{fi:vecmat}
\end{figure*}

\subsection{Fine-grained vs. Coarse-grained Decompositions} 

A relevant feature of \dc\ is that designers can flexibly decide at which level of granularity a given algorithmic solution can be decomposed into smaller parts, i.e., they might use a single constraint that performs the entire computation (coarse-grained decomposition), or many constraints each computing only a small portion of the program's state (fine-grained decomposition). In reactive scenarios, where constraints are re-evaluated selectively only on the affected portions of the input, this design choice can have implications both on memory usage and on running time. To explore these tradeoffs, we experimented with matrix benchmarks {\tt matmat} and {\tt vecmat}. For brevity, in this section we focus on {\tt vecmat}, the results for {\tt matmat} being similar.

\begin{figure}
\centerline{
\includegraphics[width=6.3cm]{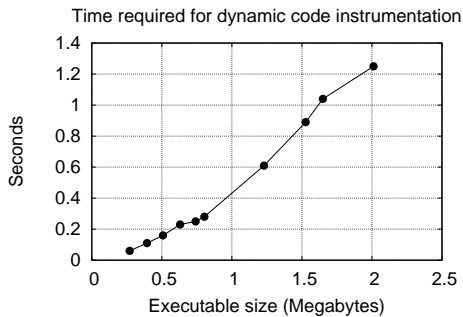}
}
\bigskip
\caption{Time required for dynamic instrumentation.}
\label{fi:instrtimes}
\end{figure}

Let $V$ be a vector of size $n$ and let $M$ be a reactive matrix of size $n\times n$. Our implementation of the vector-matrix product algorithm is {\em blocked}: constraints are associated to blocks of matrix cells, where a block is a set of consecutive cells on the same column. 
If the block size is 1, then there is one constraint per matrix cell: constraint $c_{i,j}$ is responsible of updating the $j$-th entry of the output vector with the product $V[i]*M[i][j]$. This can be done in $O(1)$ time by maintaining a local copy of the old product value and updating the result with the difference between the new value and the old one.
If the block size is $n$, then there is a constraint per matrix column: constraint $c_j$ associated with column $j$ computes the scalar product between $V$ and $M[\cdot][j]$ and updates the $j$-th entry of the output vector with the new value. The approach can be naturally adapted to deal with any block size $b\in[1,n]$.

In Figure~\ref{fi:vecmat} we report the outcome of an experiment with $n=2000$ in which we increased the block size $b$ from 1 to $n$. As shown in Figure~\ref{fi:vecmat}a, the memory usage is inversely proportional to $b$, and thus directly proportional to the number of constraints: the memory used for maintaining constraints is about half of the total  amount when $b=1$, and negligible when $b=n$. All the other components (in particular, reactive memory, shadow memory, and dependencies) do not depend on the specific block size and remain constant. Figure~\ref{fi:vecmat}b shows the effect of $b$ on the change propagation times for two different kinds of updates. For single cells updates, the time scales linearly with $b$ (axes are on a log-log scale): this confirms the intuition that if an update changes only a single cell, implementations using larger block sizes perform a lot of unnecessary work. The scenario is completely different if single updates need to change entire columns (which is a typical operation for instance in incremental graph reachability algorithms~\cite{Demetrescu+2001a}): in this case, change propagation time not only is not penalized by larger block sizes, but it is also slightly improved. This is due to the fact that larger values of $b$ yield a smaller number of constraints, which induces smaller scheduling activity. The improvement, however, is modest, suggesting that \dc's constraint scheduling overhead is modest compared to the overall work required to solve a given problem even at the finest-grained decomposition where there is one constraint per matrix cell.

\subsection{Instrumentation Overhead} 

As a final set of experiments, we have measured how instrumentation time scales as a function of the executable file size. We noticed that the performance overheads are dominated by the initial static binary code analysis phase performed during \dc's initialization, which scans the code to index memory access instructions as described in Section~\ref{ss:reactive-implem}. The times required for access violation handling and just-in-time binary code patching are negligible compared to the overall execution times in all tested applications and are not reported. The experiment was conducted by initializing \dc\ on executable files obtained by linking statically object files of increasing total size. The results are reported in Figure~\ref{fi:instrtimes} and indicate that \dc\ scales linearly, with total instrumentation times being reasonably small even for large executable sizes.

\section{Related Work}
\label{se:relatedwork}

The ability of a program to respond to modifications of its environment is a feature that has been widely explored in a large variety of settings and along rather different research lines. While this section is far from being exhaustive, we discuss some previous works that appear to be more closely related to ours.

\paragraph{GUI and Animation Toolkits.} Although dataflow programming is a general paradigm, dataflow constraints have gained popularity in the 90's especially in the creation of interactive user interfaces. Amulet~\cite{Amulet97} and its predecessor Garnet~\cite{myers90} are graphic user interface toolkits based on the dataflow paradigm. Amulet integrates a constraint solver with a prototype-instance object model implemented on top of C++, and is closely related to our work. Each object, created by making an instance of a prototype object, consists of a set of properties (e.g., appearance or position) that are stored in reactive variables, called slots. Constraints are created by assigning formulas to slots. Values of slots are accessed through a {\tt Get} method that, when invoked from inside of a formula, sets up a dependency between slots. A variety of approaches have been tested by the developers to solve constraints~\cite{Zanden01}. FRAN (Functional Reactive Animation) provides a reactive environment for composing multimedia animations through temporal modeling~\cite{WanH00}: graphical objects in FRAN use time-varying, reactive variables to automatically change their properties, achieving an animation that is function of both events and time.

\paragraph{Reactive Languages.} The dataflow model of computation can be also supported directly by programming languages. Most of them are visual languages, often used in industrial settings~\cite{Blume07}, and allow the programmer to directly manage the dataflow graph by visually putting links between the various entities. Only a few non-visual languages provide a dataflow environment, mostly for specific domains. Among them, Signal \cite{signal86} and Lustre \cite{lustre87} are dedicated to programming real-time systems found in embedded software, and SystemC~\cite{systemcbook02} is a system-level specification and design language based on C++. The data-driven Alpha language provided by the Leonardo software visualization system allows programmers to specify declarative mappings between the state of a C program an a graphical representation of its data structures~\cite{DFP00}.
Recently, Meyerovich \emph{et al.}~\cite{1640091} have introduced Flapjax, a reactive extension to the JavaScript language targeted to Web applications. Flapjax offers \emph{behaviors} (e.g., variables whose value changes are automatically propagated by the language), and \emph{event streams} (e.g., potentially infinite streams of discrete events, each of which triggers additional computations). SugarCubes~\cite{SPE:SPE218} and ReactiveML~\cite{MandelPouzet-PPDP-2005} allow reactive programming (in Java and OCAML, respectively) by relying not on operating system and runtime support, as our approach does, but rather on causality analysis and a custom interpreter/compiler. These systems, however, track dependencies between functional units, through the use of specific language constructs, such as \emph{events}, and explicit commands for generating and waiting for events. 

\paragraph{Constraint Satisfaction.} Dataflow constraints fit within the more general field of constraint programming~\cite{apt03principles}. Terms such as ``constraint propagation'' and ``constraint solving'' have often been used in papers related to dataflow since the early developments of the area~\cite{B81, Amulet97, Zanden01}. However, the techniques developed so far in dataflow programming are quite distant from those appearing in the constraint programming literature~\cite{bessiere06}. In constraint programming, relations between variables can be stated in the form of multi-way constraints, typically specified over restricted domains such as real numbers, integers, or Booleans. Domain-specific solvers use knowledge of the domain in order to forbid explicitly values or combinations of values for some variables~\cite{bessiere06}, while dataflow constraint solvers are domain-independent. Moving from early work on attribute grammars~\cite{knuth68, DemersRT81}, a variety of incremental algorithms for performing efficient dataflow constraint satisfaction have been proposed in the literature and integrated in dataflow systems such as Amulet. These algorithms are based either on a mark-sweep approach~\cite{DemersRT81,Hudson91}, or on a topological ordering~\cite{AlpernHRSZ90,Hoover87}. In contrast, \dc\ uses a priority-based approach, which allows users to customize the constraint scheduling order. Mark-sweep algorithms are preferable when the dataflow graph can change dynamically during constraint evaluation: this may happen if constraints use indirection and conditionals, and thus cannot be statically analyzed.  With both approaches, if there are cyclic dependencies between constraints, they are arbitrarily broken, paying attention to evaluate each constraint in a cycle at most once. Compared to our iterative approach, this limits the expressive power of constraints.

\paragraph{Self-adjusting Computation.} A final related area, that we have extensively discussed throughout the paper, is that of self-adjusting computation, in which programs respond to input changes by updating automatically their output. This is achieved by recording data and control dependencies during the execution of programs so that a change propagation algorithm can update the computation as if the program were run from scratch, but executing only those parts of the computation affected by changes. We refer to~\cite{AcarBBT06,HammerAC09,DBLP:conf/pldi/AcarBLTT10} for recent progress in this field.

\section{Future Work}
\label{se:conclusions}

The work presented in this paper paves the road to several further developments. Although conventional platforms offer limited support for implementing reactive memory efficiently, we believe that our approach can greatly benefit from advances in the hot field of transactional memories, which shares with us the same fundamental need for a fine-grained, highly-efficient control over memory accesses. Multi-core platforms suggest another interesting direction. Indeed, exposing parallelism was one of the motivations for  dataflow architectures, since the early developments of the area. We regard it as a challenging goal to design effective models and efficient implementations of one-way dataflow constraints in multi-core environments.

\acks

We wish to thank Umut Acar and Matthew Hammer for many enlightening discussions and for their support with \ceal. We are also indebted to Alessandro Macchioni for his contributions to the implementation of reactive memory, and to Pietro Cenciarelli and Ivano Salvo for providing useful feedback on the formal aspects of our work. 

This work was supported in part by the Italian Ministry of Education, University, and Research (MIUR) under PRIN 2008TFBWL4 national research project ``AlgoDEEP: Algorithmic challenges for data-intensive processing on emerging computing platforms''.

\bibliographystyle{abbrvnat}
\softraggedright

\end{document}